\documentclass[journal, draft,onecolumn]{IEEEtran}

\usepackage{url}
\usepackage{microtype}

\usepackage{hyperref}
\hypersetup{
  colorlinks   = true, 
  linkcolor    = black!50!brown, 
  citecolor   = black!50!brown, 
}

\usepackage{url}
\usepackage{hyperref}
\hypersetup{
  colorlinks   = true, 
  urlcolor     = black!50!red, 
  linkcolor    = black!50!brown, 
  citecolor   = black!50!brown 
}
\usepackage{todonotes}
\usepackage{amsmath}
\usepackage{amssymb}
\usepackage{amsthm}
\usepackage{verbatim}
\usepackage{bm}
\usepackage{color,graphicx,xcolor}
\usepackage{mdwtab}
\usepackage{subfigure}
\usepackage{mathtools,tikz}
\usepackage{hhline}
\usepackage{multirow}
\usepackage{pdfpages}
\usepackage{enumitem}
\usepackage{balance}

\usepackage{cite}
\usepackage{amsmath,amssymb,amsfonts,amsthm,xspace}
\usepackage{algorithmic}
\usepackage{graphicx}
\usepackage{textcomp}

\usepackage{cite}
\usepackage{amsmath,amssymb,amsfonts}
\usepackage{algorithmic}
\usepackage{graphicx}
\usepackage{textcomp}
\usepackage{xcolor}
\usepackage{amsthm}

\usepackage{pgfplots}
\pgfplotsset{compat=1.15}
\usepackage{mathrsfs}
\usetikzlibrary{arrows}
\usepackage{adjustbox}
\usepackage{booktabs}
\usepackage{authblk}

\newcommand{\nullHyp}{\mathcal{W}}
\newcommand{\altHyp}{\overline{\mathcal{W}}} 
\newcommand{\Wbar}{\overline{W}}
\newcommand{\Ubar}{\overline{U}}
\newcommand{\scriptS}{\mathcal{S}}
\newcommand{\scriptSbar}{\bar{\mathcal{S}}}
\newcommand{\scriptX}{\mathcal{X}}
\newcommand{\scriptY}{\mathcal{Y}}
\newcommand{\domXY}{{\mathcal{X}}^n \times {\mathcal{Y}}^n}
\newcommand{\typeOneErrShared}{{\alpha}_n^{\textup{sh}}}
\newcommand{\typeTwoErrShared}{{\beta}_{n}^{\epsilon,\textup{sh}}}
\newcommand{\testShared}{f_{\textup{sh}}}
\newcommand{\advExpShared}{{\mathcal{E}}_{\textup{sh}}^{\epsilon}}
\newcommand{\advExpDet}{{\mathcal{E}}_{\textup{det}}^{\epsilon}}
\newcommand{\advExpPriv}{{\mathcal{E}}_{\textup{priv}}^{\epsilon}}
\newcommand{\advExp}{\mathcal{E}_{\textup{adv}}^{\epsilon}}
\newcommand{\eps}{\epsilon}
\newcommand{\Qshared}{Q_{\textup{sh}}}
\newcommand{\Qbarshared}{\bar{Q}_{\textup{sh}}}
\newcommand{\convHull}{\textup{conv}}
\newcommand{\bern}{\textup{Bern}}
\newcommand{\DOptSh}{D_{\textup{sh}}^{*}}
\newcommand{\DOptShBar}{\overline{D}_{\textup{sh}}^{*}}
\newcommand{\DOptDet}{D_{\textup{det}}^{*}}
\newcommand{\DOptDetBar}{\overline{D}_{\textup{det}}^{*}}
\newcommand{\DOptPriv}{D_{\textup{priv}}^{*}}
\newcommand{\DOptPrivBar}{\overline{D}_{\textup{priv}}^{*}}
\newcommand{\typeOneErrRest}{\tilde{\alpha}_{n}}
\newcommand{\typeTwoErrRest}{\tilde{\beta}_{n}}

\newcommand{\Prob}{\textup{Pr}}
\newcommand{\phiLarge}{\Phi_t}
\newcommand{\phiSmall}{\phi_t}
\newcommand{\phiStarShared}{\phi^{*}_{\textup{sh}}}
\newcommand{\Qtilt}{Q_{\textup{tilt}}}
\newcommand{\Qdet}{Q_{\textup{det}}}
\newcommand{\Qbardet}{\bar{Q}_{\textup{det}}}

\newcommand{\testDet}{f_{\textup{det}}}
\newcommand{\typeOneErrDet}{{\alpha}_n^{\textup{det}}}
\newcommand{\typeTwoErrDet}{{\beta}_{n}^{\epsilon,\textup{det}}}
\newcommand{\Qpriv}{Q_{\textup{priv}}}
\newcommand{\Qbarpriv}{\bar{Q}_{\textup{priv}}}
\newcommand{\testPriv}{f_{\textup{priv}}}
\newcommand{\typeOneErrPriv}{{\alpha}_n^{\textup{priv}}}
\newcommand{\typeTwoErrPriv}{{\beta}_{n}^{\epsilon,\textup{priv}}}

\newcommand{\sX}{\mathcal{X}}

\allowdisplaybreaks

\usepackage{pgffor}
\foreach \x in {a,...,z}{%
\expandafter\xdef\csname vec\x \endcsname{\noexpand\ensuremath{\noexpand\bm{\x}}}
}

\foreach \x in {A,...,Z}{%
\expandafter\xdef\csname vec\x \endcsname{\noexpand\ensuremath{\noexpand\bm{\x}}}
}

\foreach \x in {A,...,Z}{%
\expandafter\xdef\csname c\x \endcsname{\noexpand\ensuremath{\noexpand\mathcal{\x}}}
}

\foreach \x in {A,...,Z}{%
\expandafter\xdef\csname bb\x \endcsname{\noexpand\ensuremath{\noexpand\mathbb{\x}}}
}

\newcommand{\defineqq}{\ensuremath{\stackrel{\textup{\tiny def}}{=}}}

\newcommand{\inb}[1]{\left\{#1\right\}}
\newcommand{\inp}[1]{\left(#1\right)}
\newcommand{\insq}[1]{\left[#1\right]}
\newcommand{\inps}[1]{\left|#1\right|^+}
\newcommand{\inl}[1]{\left|#1\right|}

\newcommand{\bec}[1]{\ensuremath{\mathsf{BEC}(#1)}}

\newcommand{\red}{\textcolor{red}}

\newcommand{\bpvtexp}{\ensuremath{\cE^{\epsilon}_{\textup{pvt}}(\nullHyp,\altHyp)}}

\theoremstyle{definition}
\newtheorem{example}{Example}

\newtheorem{theorem}{Theorem}
\newtheorem{corollary}{Corollary}
\newtheorem{lemma}{Lemma}
\newtheorem{remark}{Remark}
\newtheorem{definition}{Definition}

\def\BibTeX{{\rm B\kern-.05em{\sc i\kern-.025em b}\kern-.08em
    T\kern-.1667em\lower.7ex\hbox{E}\kern-.125emX}}

\IEEEoverridecommandlockouts

\usepackage[normalem]{ulem}
\begin{document}

\title{Hypothesis Testing for Adversarial Channels: Chernoff-Stein Exponents \thanks{This work was presented in part at the 2023 IEEE International Symposium on Information Theory.\\
E. Modak, N. Sangwan and V. M. Prabhakaran were supported by DAE under project no. RTI4001. N. Sangwan was additionally supported by the TCS Foundation through the TCS Research Scholar Program.  The work of M. Bakshi was supported by the National Science Foundation under Grant No. CCF-2107526. The work of B. K. Dey was supported in part by Bharti Centre for Communication in IIT Bombay. V. M. Prabhakaran was additionally supported by SERB through project MTR/2020/000308.}}

\author[1]{Eeshan Modak}
\author[2]{Neha Sangwan}
\author[3]{Mayank Bakshi}
\author[4]{Bikash Kumar Dey}
\author[1]{Vinod M. Prabhakaran}
\affil[1]{Tata Institute of Fundamental Research, Mumbai, India}
\affil[2]{University of California, San Diego, CA, USA}
\affil[3]{Arizona State University, Tempe, AZ, USA}
\affil[4]{Indian Institute of Technology Bombay, Mumbai, India}

\maketitle

\begin{abstract}

Consider the following binary hypothesis testing
problem: Associated with each hypothesis is a set of channels. A transmitter,
without knowledge of the hypothesis, chooses the inputs to the
channel. Given the hypothesis, from the set associated with the hypothesis, an
adversary chooses channels, one for each element of the input vector. Based on
the channel outputs, a detector attempts to distinguish between the hypotheses.
For the fixed-length setting, we study the Chernoff-Stein exponent for the cases where the transmitter (i) is deterministic, (ii) may privately randomize, and (iii) shares randomness with
the detector that is unavailable to the adversary. It turns out that while a
memoryless transmission strategy is optimal under shared randomness, it may be
strictly suboptimal when the transmitter only has private randomness. We also study the sequential version of this problem in each of the three settings and show that both the Chernoff-Stein exponents can be simultaneously achieved.
\end{abstract}


\section{Introduction}\label{sec:introduction}

We study the binary hypothesis testing problem for arbitrarily varying channels (AVC)~\cite{blackwell1960capacities}. Associated with each hypothesis is a set of channels. All channels have the same input and output alphabets. The transmitter, without knowledge of the hypothesis, chooses the vector of inputs to the channel. Given the hypothesis, the adversary chooses a vector of channels where each element belongs to the set of channels associated with the hypothesis. The detector observes the outputs resulting from applying the inputs chosen by the transmitter element-wise independently to the channels selected by the adversary. It then makes a decision on the hypothesis. The adversary is aware of the strategy of the transmitter and detector, but not necessarily the choice of channel inputs. 

In simple binary hypothesis testing~\cite{chernoff1952measure,hoeffding1965asymptotically} the goal is to distinguish between two distributions (sources), say $H_0: p$ and $H_1: q$ from $n$ independent and identically distributed (i.i.d.) observations from the source. The Chernoff-Stein lemma~\cite[Theorem~11.8.3]{cover1999elements} states that for a fixed false alarm (type-1 error) probability, the optimal missed detection (type-2 error) probability decays exponentially in $n$ with the exponent given by the relative entropy $D(p\|q)$ between the distributions. The test which achieves this exponent is a likelihood ratio test. When the detector is allowed to observe a variable number of samples, Wald and Wolfowitz \cite{wald1948optimum} showed that the pair of exponents $(D(q\|p),D(p\|q))$ can be simultaneously achieved by the sequential probability ratio test (SPRT) with appropriate thresholds.

A variation on this problem is where each observation is from an arbitrarily varying source~\cite{strassen1964me}. There is a set of distributions associated with each hypothesis, say $H_0: \mathcal{P}$ and $H_1: \mathcal{Q}$. Given a hypothesis, the observations are independent, but each observation could be arbitrarily distributed according to any one of the distributions belonging to the set of distributions corresponding to the hypothesis. We may view the choice of distribution as being made by an adversary who is aware of the detection scheme used. Fangwei and Shiyi~\cite{fangwei1996hypothesis} studied this problem where the adversary's choice may be stochastic but unaware of past observations. They showed that when the sets are closed and convex, the Chernoff-Stein exponent for this problem is given by $\min \limits_{p \in \mathcal{P},q \in \mathcal{Q}}D(p\|q)$. Brand\~{a}o, Harrow, Lee, and Peres~\cite{brandao2020adversarial} strengthened this result by showing that the above exponent remains unchanged even when the adversary is adaptive, i.e. it has feedback of the past observations and may use this to choose the distribution of the next observation. In both cases, the optimal test is a likelihood ratio test with respect to the closest pair of distributions between the two sets.

In another variation on the binary hypothesis testing problem, instead of distinguishing between sources, the objective is to distinguish between two channels (say $H_0: W$ and $H_1: \Wbar$) with the same input (say $\mathcal{X}$) and output alphabets (say $\mathcal{Y}$)~\cite{blahut1974hypothesis,hayashi2009discrimination}. Here, a transmitter, which is unaware of the hypothesis, may choose the inputs to the channels. It was shown that the optimal Chernoff-Stein error exponent can be attained using a deterministic transmission strategy, which sends the input letter for which the relative entropy between the channel output distributions under the two hypotheses is maximized (i.e. most discriminating symbol). The optimal exponent is given by $\max \limits_{x \in \mathcal{X}} D(W(.|x)\|\Wbar(.|x))$. Hayashi~\cite{hayashi2009discrimination} further showed that feedback does not improve the optimal error exponent in the adaptive case where the transmitter has feedback of the channel output. The optimal scheme is to send the most discriminating symbol during all channel uses and then performing a likelihood ratio test on the channel outputs. Polyanskiy and Verd\'u\cite{polyanskiy2011binary} considered the same problem with variable-length transmissions and showed that the pair of Chernoff-Stein exponents ($\max \limits_{x \in \mathcal{X}} D(\Wbar(.|x)\|W(.|x))$, $\max \limits_{x \in \mathcal{X}} D(W(.|x)\|\Wbar(.|x))$) can be simultaneously achieved. 

We consider the problem of distinguishing between two arbitrarily varying channels (say $H_0:\nullHyp$ and $H_1:\altHyp$). As in~\cite{hayashi2009discrimination}, a transmitter chooses the inputs to the channels. The sequence of channel states is (possibly randomly) chosen by an adversary who knows the strategy employed by the transmitter and the detector but not any shared or private randomness available to them. We first examine this problem in the fixed-length setting. We study three different cases based on the nature of randomness hidden from the adversary\footnote{We allow the adversary to randomize in all cases.}: (i) randomness shared between transmitter and detector (Section~\ref{sec:shared}), (ii) deterministic schemes (Section~\ref{sec:deterministic}), and (iii) private randomness at the transmitter (Section~\ref{sec:private}).  
We also comment on the role of adaptivity both of the transmitter and of the adversary (Section~\ref{sec:adaptive}).

\begin{table*}[]
	\centering
	\begin{tabular}{lcl} \toprule 
		& Chernoff-Stein exponent & Condition for the exponent to be non-zero \\
		\midrule
		\textup{Shared randomness} & $\sup \limits_{P_X} \min \limits_{U \in \convHull(\nullHyp), \Ubar \in \convHull(\altHyp)} D(U\|\Ubar|P_{X})$ & $\convHull(\nullHyp)\cap \convHull(\altHyp)=\emptyset$ \\
		\textup{Deterministic transmitter} & $\max \limits_{x} \min \limits_{U_x \in \convHull(\nullHyp_x), \Ubar_x \in \convHull(\altHyp_x)} D(U_x \| \Ubar_x)$ & $\convHull(\nullHyp_x)\cap \convHull(\altHyp_x)=\emptyset$ for some $x$ \\
		\textup{Private randomness} & Open (see Theorem~\ref{thm:pvt_rand}) & $\convHull(\nullHyp)\cap \convHull(\altHyp)=\emptyset$ and $(\nullHyp,\altHyp)$ is not trans-symmetrizable \\ \bottomrule
	\end{tabular}
	\label{tab:summary}
\end{table*}

In the case where randomness is shared, we show that the optimal Chernoff-Stein exponent is given by
\begin{equation*}
    \DOptSh := \sup \limits_{P_X} \min \limits_{\substack{U \in \convHull(\nullHyp) \\ \Ubar \in \convHull(\altHyp)}} D(U\|\Ubar|P_{X})
\end{equation*}
where $\convHull(\nullHyp)$ and $\convHull(\altHyp)$ are the convex hulls of the channel sets $\nullHyp$ and $\altHyp$ respectively. In contrast to \cite{hayashi2009discrimination}, randomness is necessary in general in this setting to achieve the optimal exponent. In line with their work, feedback (to the transmitter or adversary) does not change the optimal exponent. We observe that if the transmitter sends input symbols i.i.d. according to $P_X$, the problem reduces to detecting arbitrarily varying sources studied in \cite{fangwei1996hypothesis}, \cite{brandao2020adversarial}. The achievability of the exponent follows from this. The converse follows from the converse to the channel discrimination problem \cite{hayashi2009discrimination} by fixing an i.i.d. adversary strategy. While the conference version of this paper was under review, a work by Bergh, Datta and Salzmann \cite{bergh2023composite} that studies binary composite classical and quantum channel discrimination appeared. Their result in the context where the two hypotheses are convex sets of classical channels \cite[Theorem~13]{bergh2023composite} is identical to Theorem~\ref{thrm:shared_weak} (Section~\ref{sec:shared}).

In a similar vein, we show that the optimal exponent for the deterministic case is given by 
\begin{equation*}
    \DOptDet := \sup \limits_{x} \min \limits_{\substack{U_x \in \convHull(\nullHyp_x) \\ \Ubar_x \in \convHull(\altHyp_x)}} D(U_x\|\Ubar_x)
\end{equation*}
where $\convHull(\nullHyp_x)$ (resp. $\convHull(\altHyp_x)$) is the convex hull of the channel output distributions under $H_0$ (resp. $H_1$) when the input symbol is $x$ and $U_x(.)=U(.|x), \Ubar_x(.)=\Ubar(.|x)$. This holds true even when both the transmitter and the adversary have feedback. In both these cases, a memoryless transmission strategy turns out to be optimal. 

Interestingly, the optimality of the memoryless strategy does not extend to the private randomness case. In this case, the transmitter has randomness which is unknown to the adversary, but shares no randomness with the detector. A memoryless strategy can help us achieve $\sup_{P_{X}} \min_{\substack{Q_{Y} \in \mathcal{Q} \\ \bar{Q}_{Y} \in \bar{\mathcal{Q}}}} D(Q_{Y}\|\bar{Q}_{Y})$, where $\mathcal{Q}$ (resp. $\bar{\mathcal{Q}}$) is the set of (single-letter) channel output distributions that can be induced by the adversary when the input is distributed as $P_X$ under hypothesis $H_0$ (resp. $H_1$). We show that not only is this not the optimal exponent, this expression can evaluate to zero even when the optimal exponent is positive (Example~\ref{ex:role_of_memory}, Section~\ref{sec:private}). We characterize the conditions under which the exponent is positive. We also give a lower bound on the exponent using some ideas from codes for arbitrarily varying channels. Our model with private randomness is related to~\cite{chaudhuri2021compound} and is discussed in Section~\ref{sec:private} which considered communication rates and feasibility but not error exponents.

We also study the sequential version of this problem (Section~\ref{sec:seq_shared}). In this case, transmissions can be of variable length (with constraints on the expected length), and the detector's decision is based on a stopping rule. In each of the three settings of randomness, we show that the pair of optimal (fixed length) Chernoff-Stein exponents can be simultaneously achieved. These results are along the lines of \cite{wald1948optimum}, \cite{polyanskiy2011binary}. The achievability is based on a lemma which shows how to combine fixed length schemes to construct the desired sequential test (refer Lemma~\ref{lemma:seq_test}). Our scheme is along the lines of the two-phase sequential tests studied by Chernoff \cite{chernoff1959sequential}, Kiefer and Sacks \cite{kiefer1963asymptotically} and Naghshvar and Javidi \cite{naghshvar2013active}.

Our main contributions are the following.
\begin{itemize}
    \item We study the testing problem between two AVCs. We give an exact characterization of the Chernoff-Stein exponent for the shared randomness (Theorem~\ref{thrm:shared_weak}, Section~\ref{sec:shared}) and deterministic case (Theorem~\ref{thrm:det_weak}, Section~\ref{sec:deterministic}). For the private randomness case, we get an achievable exponent which in general can be sub-optimal (Theorem~\ref{thm:pvt_rand}, Section~\ref{sec:private})
    \item We observe that i.i.d. transmission strategies are optimal for the shared randomness case but are sub-optimal for the private randomness case in general as demonstrated in  Example~\ref{ex:role_of_memory}.
    \item We show that randomness helps to boost the exponent unlike the non-adversarial channel discrimination problem~\cite{hayashi2009discrimination}. As in the case of \cite{hayashi2009discrimination}, we observe that feedback does not help to increase the exponent in the shared randomness case.
    \item Finally, we also study the sequential version of the  problem, and show that both the Chernoff-Stein exponents of the fixed length problem can be simultaneously achieved in the sequential version (Theorems~\ref{thrm:shared_weak_seq}, \ref{thrm:det_weak_seq} and \ref{thrm:priv_weak_seq}, Section~\ref{sec:seq_shared}).
\end{itemize}

\section{Problem Setup} \label{sec:problem_setup}
Let $\scriptX$ and $\scriptY$ be finite sets. A discrete memoryless channel $W(.|.)$ takes an input symbol $x \in \mathcal{X}$ and outputs a symbol $y \in \mathcal{Y}$ with probability $W(y|x)$. Consider two finite sets of channels $\nullHyp = \{W(.|.,s):s \in \scriptS\}$, $\altHyp = \{\Wbar(.|.,\bar{s}):\bar{s} \in \scriptSbar\}$ which map $\mathcal{X}$ to $\mathcal{Y}$. The goal is to distinguish between the two sets of channels. In particular, we study the asymmetric hypothesis test between the null hypothesis $H_{0}:\nullHyp$ and the alternative hypothesis $H_{1}:\altHyp$. There are three entities involved: (a) the transmitter, (b) the adversary, and (c) the detector. The transmitter is unaware of which hypothesis has been realized and chooses the input symbols. The adversary, depending on which hypothesis is realized, chooses the state symbols (from $\scriptS$ under $H_0$ and $\scriptSbar$ under $H_1$). The detector decides between $H_0$ and $H_1$ based on everything it knows. We consider three different settings. In each of the settings, we seek to characterize the Chernoff-Stein exponent of the problem.

\subsection{Shared Randomness}
In this setting, the transmitter and detector share randomness which is unknown to the adversary. The input $X^n$ to the channel, which is a function of this randomness, is known to the detector. For a transmitter strategy $P_{X^n}$ and a pair of adversary strategies $P_{S^n}$ and $P_{\bar{S}^n}$, the distribution induced on $\domXY$ under $H_0$ is given by\footnote{For compactness of notation, in \eqref{eqn:q_shared} the dependence of $\Qshared^{n}$ on the transmission strategy $P_{X^n}$ and the adversary strategy $P_{S^n}$ is suppressed. And in \eqref{eqn:q_det} the dependence of $\Qdet^{n}$ on the transmission strategy $x^n$ and the adversary strategy $P_{S^n}$ is suppressed.\label{foot1}}
\begin{equation} \label{eqn:q_shared}
\Qshared^{n}(x^n,y^n) = \sum \limits_{s^n \in \scriptS^n} P_{X^n}(x^n)P_{S^n}(s^n)\prod \limits_{i=1}^{n}W(y_i|x_i,s_i).
\end{equation}
A similar expression is obtained for $\Qbarshared^{n}$ under $H_1$ where instead of $P_{S^n}$ and $W$ we have $P_{\bar{S}^n}$ and $\overline{W}$ respectively. 
The detector uses a (possibly privately randomized) decision rule $\testShared : \domXY \rightarrow \{ 0, 1\}$. Let $A_{n}$ be the (possibly random) acceptance region for $H_0$, i.e., $A_{n} = \{ (x^n,y^n) \in \domXY : \testShared(x^n,y^n) = 0 \}$. A scheme for the shared randomness case is given by a pair of transmission strategy and detection rule $(P_{X^n}, \testShared)$. For a given scheme, the type-I error is given by 
\begin{equation*}
\typeOneErrShared = \sup \limits_{P_{S^n}} \mathbb{E} \left [ \Qshared^{n}(A_{n}^c) \right ],
\end{equation*}
where the expectation is over the random choice of $A_{n}$.
For $\eps>0$, when the type-I error $\typeOneErrShared$ is at most $\eps$, the optimal type-II error is given by
\begin{equation*}
\typeTwoErrShared \defineqq \inf \limits_{P_{X^n}} \inf \limits_{A_{n}:\typeOneErrShared \le \epsilon} \sup \limits_{P_{\bar{S}^n}} \mathbb{E} \left[ \Qbarshared^{n}(A_{n}) \right],
\end{equation*}
where the expectation is over the random $A_n$ set by the inner $\inf$.
The Chernoff-Stein exponent is then defined to be
\begin{equation*}
\advExpShared(\nullHyp,\altHyp) \defineqq \liminf  \limits_{n \rightarrow \infty} -\frac{1}{n} \log \typeTwoErrShared, \quad\eps > 0.
\end{equation*}

\subsection{Deterministic}
In this setting, the transmitter strategy is completely deterministic and is defined by a fixed tuple $(x_1,x_2,\dots,x_n)$. For this transmission strategy and an adversary strategy $P_{S^n}$, the distribution on $\scriptY^n$ under $H_0$ is given by${}^{\thefootnote}$
\begin{equation} \label{eqn:q_det}
\Qdet^{n}(y^n) = \sum \limits_{s^n \in \scriptS^n}P_{S^n}(s^n)\prod \limits_{i=1}^{n}W(y_i|x_i,s_i).
\end{equation}
A similar expression is obtained for $\Qbardet^{n}$ under $H_1$ where instead of $P_{S^n}$ and $W$ we have $P_{\bar{S}^n}$ and $\overline{W}$ respectively. The decision rule used by the detector is specified by $\testDet : \mathcal{Y}^n \rightarrow \{ 0, 1\}$. Let $A_{n}$ be the (possibly random) acceptance region for $H_0$, i.e., $A_{n} = \{ y^n \in \mathcal{Y}^n : \testDet(y^n) = 0 \}$. A scheme for the deterministic case is given by a pair of transmission strategy and detection rule $(x^n, \testDet)$. For a given scheme, the type-I error is given by 
\begin{equation*}
\typeOneErrDet = \sup \limits_{P_{S^n}} \mathbb{E} \left [ \Qdet^{n}(A_{n}^c) \right],
\end{equation*}
where the expectation is over the random choice of $A_{n}$. For $\eps>0$, when the type-I error $\typeOneErrDet$ is at most $\eps$, the optimal type-II error is given by
\begin{equation*}
\typeTwoErrDet \defineqq \inf \limits_{x^n} \inf \limits_{A_{n}:\typeOneErrDet \le \epsilon} \sup \limits_{P_{\bar{S}^n}} \mathbb{E} \left[ \Qbardet^{n}(A_{n}) \right],
\end{equation*}
where the expectation is over the random $A_n$ set by the inner $\inf$. The Chernoff-Stein exponent is then defined to be
\begin{equation*}
\advExpDet(\nullHyp,\altHyp) \defineqq \liminf  \limits_{n \rightarrow \infty} -\frac{1}{n} \log \typeTwoErrDet, \quad\eps > 0.
\end{equation*}

\subsection{Private Randomness}
We finally consider the case where the transmitter may choose the channel input $X^n$ randomly, but the realization of $X^n$ is unavailable to the detector and the adversary. For a transmitter strategy $P_{X^n}$ and an adversary strategy $P_{S^n}$, the distribution induced on $\mathcal{Y}^n$ under $H_0$ is given by \footnote{For compactness of notation, in \eqref{eqn:q_priv} the dependence of $\Qpriv^{n}$ on the transmission strategy $P_{X^n}$ and the adversary strategy $P_{S^n}$ is suppressed.}
\begin{equation} \label{eqn:q_priv}
\Qpriv^{n}(y^n) = \sum_{\substack{x^n \in \mathcal{X}^n \\ s^n \in \scriptS^n}} P_{X^n}(x^n)P_{S^n}(s^n)\prod \limits_{i=1}^{n}W(y_i|x_i,s_i).
\end{equation}
A similar expression is obtained for $\Qbarpriv^{n}$ under $H_1$ where instead of $P_{S^n}$ and $W$ we have $P_{\bar{S}^n}$ and $\overline{W}$ respectively. The decision rule used by the detector is specified by $\testPriv : \mathcal{Y}^n \rightarrow \{ 0, 1\}$. Let $A_{n}$ be the (possibly random) acceptance region for $H_0$, i.e., $A_{n} = \{ y^n \in \mathcal{Y}^n : \testDet(y^n) = 0 \}$. A scheme for the private randomness case is given by a pair of transmission strategy and detection rule $(P_{X^n}, \testPriv)$. For a given scheme, the type-I error is given by 
\begin{equation*}
\typeOneErrPriv = \sup \limits_{P_{S^n}} \mathbb{E} \left [ \Qpriv^{n}(A_{n}^c) \right],
\end{equation*}
where the expectation is over the random choice of $A_{n}$. For $\eps>0$, when the type-I error $\typeOneErrPriv$ is at most $\eps$, the optimal type-II error is given by
\begin{equation*}
\typeTwoErrPriv \defineqq \inf \limits_{P_{X^n}} \inf \limits_{A_{n}:\typeOneErrPriv \le \epsilon} \sup \limits_{P_{\bar{S}^n}} \mathbb{E} \left[ \Qbarpriv^{n}(A_{n}) \right],
\end{equation*}
where the expectation is over the random $A_n$ set by the inner $\inf$. The Chernoff-Stein exponent is then defined to be
\begin{equation*}
\advExpPriv(\nullHyp,\altHyp) \defineqq \liminf  \limits_{n \rightarrow \infty} -\frac{1}{n} \log \typeTwoErrPriv, \quad\eps > 0.
\end{equation*}

We also study the sequential versions of the above problems. We discuss it separately in Section~\ref{sec:seq_shared}. We now present the results for each of the above setting.

\begin{figure}
    \resizebox{0.95\columnwidth}{!}{

\begin{tikzpicture}[x=0.75pt,y=0.75pt,yscale=-1,xscale=1]

\draw   (158.5,142.37) -- (227.44,142.37) -- (227.44,168) -- (158.5,168) -- cycle ;
\draw   (426.99,141.32) -- (478.5,141.32) -- (478.5,168.53) -- (426.99,168.53) -- cycle ;
\draw   (294.2,231.47) -- (353.5,231.47) -- (353.5,257) -- (294.2,257) -- cycle ;
\draw   (286.5,137) .. controls (286.5,130.37) and (291.87,125) .. (298.5,125) -- (350.5,125) .. controls (357.13,125) and (362.5,130.37) .. (362.5,137) -- (362.5,173) .. controls (362.5,179.63) and (357.13,185) .. (350.5,185) -- (298.5,185) .. controls (291.87,185) and (286.5,179.63) .. (286.5,173) -- cycle ;
\draw    (226.5,156) -- (282.76,156.28) ;
\draw [shift={(284.76,156.29)}, rotate = 180.28] [color={rgb, 255:red, 0; green, 0; blue, 0 }  ][line width=0.75]    (10.93,-3.29) .. controls (6.95,-1.4) and (3.31,-0.3) .. (0,0) .. controls (3.31,0.3) and (6.95,1.4) .. (10.93,3.29)   ;
\draw    (363.2,154.61) -- (424.5,154.99) ;
\draw [shift={(426.5,155)}, rotate = 180.35] [color={rgb, 255:red, 0; green, 0; blue, 0 }  ][line width=0.75]    (10.93,-3.29) .. controls (6.95,-1.4) and (3.31,-0.3) .. (0,0) .. controls (3.31,0.3) and (6.95,1.4) .. (10.93,3.29)   ;
\draw    (323.27,231.47) -- (323.27,187.22) ;
\draw [shift={(323.27,185.22)}, rotate = 90] [color={rgb, 255:red, 0; green, 0; blue, 0 }  ][line width=0.75]    (10.93,-3.29) .. controls (6.95,-1.4) and (3.31,-0.3) .. (0,0) .. controls (3.31,0.3) and (6.95,1.4) .. (10.93,3.29)   ;

\draw (161.29,150.51) node [anchor=north west][inner sep=0.75pt]  [font=\footnotesize,rotate=-359.91] [align=left] {Transmitter};
\draw (428.61,150.09) node [anchor=north west][inner sep=0.75pt]  [font=\footnotesize] [align=left] {Detector};
\draw (295.9,240.24) node [anchor=north west][inner sep=0.75pt]  [font=\footnotesize] [align=left] {Adversary};
\draw (297.83,128.12) node [anchor=north west][inner sep=0.75pt]  [font=\normalsize]  {$H_{0} :\ \mathcal{W} $};
\draw (298.43,159.85) node [anchor=north west][inner sep=0.75pt]  [font=\normalsize]  {$H_{1} :\overline{\mathcal{W}}$};
\draw (315.03,145.41) node [anchor=north west][inner sep=0.75pt]   [align=left] {or};
\draw (245,131.4) node [anchor=north west][inner sep=0.75pt]    {$x^{n}$};
\draw (382,129.4) node [anchor=north west][inner sep=0.75pt]    {$y^{n}$};
\draw (333,201.4) node [anchor=north west][inner sep=0.75pt]    {$s^{n} ,\overline{s}^{n}$};

\end{tikzpicture}
}
    \caption{Each hypothesis is a AVC controlled by an adversary. The transmitter sends a vector of inputs $x^n$. The adversary sends a vector of states ($s^n$ under $H_0$ and $\bar{s}^n$ under $H_1$). The detector observes the vector of outputs $y^n$.}
    \label{fig:fig2}
\end{figure}

\section{Shared Randomness} \label{sec:shared}
Let $\convHull(\nullHyp)$ and $\convHull(\altHyp)$ be the convex hulls of the channel sets $\nullHyp$ and $\altHyp$ respectively. i.e.,
\begin{equation*}
\convHull(\nullHyp) \defineqq \inb{\sum \limits_{s \in \scriptS}P_{S}(s)W(.|.,s) : P_{S} \in \Delta_{\scriptS}},
\end{equation*}
where $\Delta_{\scriptS}$ is the set of all probability distributions over $\scriptS$.
$\convHull(\altHyp)$ is defined similarly with $\bar{S},\Wbar$ instead of $S,W$.
Let
\begin{equation} \label{eqn:DoptShDef}
\DOptSh \defineqq \sup \limits_{P_X} \min_{\substack{U \in \convHull(\nullHyp) \\ \Ubar \in \convHull(\altHyp)}} D(U\|\Ubar|P_{X}).
\end{equation}
Since $\convHull(\nullHyp)$, $\convHull(\altHyp)$ are closed, convex sets and $D(.\|.)$ is lower semi-continuous, the minimum exists.

\begin{theorem}[] \label{thrm:shared_weak}
	Let $\nullHyp$ and $\altHyp$ be two sets of discrete memoryless channels which map $\mathcal{X}$ to $\mathcal{Y}$. For any $\eps \in (0,1)$, we have
	\begin{equation} \label{eqn:shared_weak}
	\DOptSh \le \advExpShared(\nullHyp,\altHyp) \le \frac{\DOptSh}{1-\eps}.
	\end{equation}
\end{theorem}

\begin{proof}
	
	{\em Achievability ($\advExpShared(\nullHyp,\altHyp) \ge \DOptSh $):} For this proof, we consider the case where $W(y|x)>0, \Wbar(y|x)>0$ for all $x \in \mathcal{X}, y \in \mathcal{Y}$ for each channel $W \in \nullHyp$, $\Wbar \in \altHyp$. If this assumption is not satisfied, it can be dealt with using the idea in \cite[Lemma 3]{brandao2020adversarial}. It involves discarding the actual observations with a small probability and instead sampling from a uniform distribution on $\mathcal{Y}$. The compactness of the set of probability distributions on $\mathcal{Y}$ and lower semi-coninuity of KL divergence implies that the Chernoff-Stein exponent of the modified problem approaches that of the original problem. We argue the achievability for the (stronger) adaptive adversary who has access to previous channel inputs and outputs.
	The transmitter transmits $X^n$ chosen i.i.d. according to $P_X$ using the shared randomness. This reduces the problem to the adversarial hypothesis testing problem studied in \cite{brandao2020adversarial}. For any fixed choice of $P_X$, invoking \cite[Theorem 2]{brandao2020adversarial} (refer Appendix~\ref{app:prelim}) with $\mathcal{P}=\{ P_X U: U \in \convHull(\nullHyp) \}$ and $\mathcal{Q} = \{ P_X \Ubar: \Ubar \in \convHull(\altHyp) \}$,
	\begin{equation*}
	\advExpShared(\nullHyp,\altHyp) \ge \min_{\substack{U \in \convHull(\nullHyp) \\ \Ubar \in \convHull(\altHyp)}} D(U\|\Ubar|P_{X}).
	\end{equation*}
	Optimizing over $P_X$ completes the proof of achievability.
	
	{\em Weak Converse ($\advExpShared(\nullHyp,\altHyp) \le \frac{\DOptSh}{1-\eps}$):}
    We show this converse result for an adaptive transmitter who has feedback of the outputs. Fix the following adversarial strategy: i.i.d. $P_S$ under $H_0$ and i.i.d. $P_{\bar{S}}$ under $H_1$. Let $U \in \convHull(\nullHyp), \Ubar \in \convHull(\altHyp)$ be the induced effective channels, i.e $U(y|x)=\sum_{s \in \mathcal{S} } P_S(s)W(y|x,s)$ and $\Ubar(y|x) = \sum_{\bar{s} \in \bar{\mathcal{S}} } P_{\bar{S}}(\bar{s})\Wbar(y|x,\bar{s})$. This reduces the problem to the one studied in \cite[Section VI]{hayashi2009discrimination}. We now invoke their weak converse argument.
    \begin{align*}
        \advExpShared(\nullHyp,\altHyp) &\le \frac{\max_{x} D(U(.|x) \| \Ubar(.|x))}{1-\eps} \\
        &= \frac{\sup_{P_X} D(U\|\Ubar|P_X)}{1-\eps}.
    \end{align*}
    We now choose the best adversarial strategy. Thus, we have
    \begin{equation*}
        \advExpShared(\nullHyp,\altHyp) \le \frac{\min \limits_{U,\Ubar} \sup \limits_{P_X} D(U\|\Ubar|P_X)}{1-\eps}
    \end{equation*}
    Using \cite[Lemma 14]{brandao2020adversarial}, we can change the order of $\min$ and $\sup$. This completes the converse argument. 
\end{proof}

An approach of choosing a memoryless (not necessarily i.i.d.) adversary strategy also allows us to use the proof technique of~ \cite[Section VI]{hayashi2009discrimination},\cite{nagaoka2005strong} to obtain the following strong converse (see Appendix~\ref{app:shared_strong} for a proof).
For  distributions $\mu_{XY}, \nu_{XY}$ on $\scriptX \times \scriptY$ and $t \in \mathbb{R}$, let 
\begin{align*}
\phiSmall(\mu_{X} \| \nu_{X}) & \defineqq \log \left[ \sum_{\scriptX}\mu_{X}^{1-t}\nu_{X}^{t} \right] \\
\phiSmall(\mu_{Y|X} \| \nu_{Y|X} | \mu_{X}) &\defineqq \log \mathop{\mathbb{E}}_{X \sim \mu_{X}} \left[ \sum_{\scriptY}\mu_{Y|X}^{1-t}\nu_{Y|X}^{t} \right].
\end{align*}
\begin{theorem}\label{thrm:shared_strong}
	If 
	\begin{align}
	    \lim\limits_{t \rightarrow 0^{-}} \sup \limits_{P_X} \inf_{\substack{U \in \convHull(\nullHyp) \\ \Ubar \in \convHull(\altHyp)}} \frac{\phi_{t}(U \| \Ubar | P_{X})}{-t} = \sup \limits_{P_X} \inf_{\substack{U \in \convHull(\nullHyp) \\ \Ubar \in \convHull(\altHyp)}} \lim\limits_{t \rightarrow 0^{-}} \frac{\phi_{t}(U \| \Ubar | P_{X})}{-t}, \label{eqn:shared_conv_assumption}
	\end{align}
	then
	\begin{equation}
	\advExpShared(\nullHyp,\altHyp) = \DOptSh.
	\end{equation}
\end{theorem}

The following theorem characterizes the pairs of $(\nullHyp,\altHyp)$ for which $\advExpShared>0$.
\begin{theorem}\label{thm:shared_characterization}
$\advExpShared(\nullHyp,\altHyp)>0 \iff \convHull(\nullHyp)\cap \convHull(\altHyp)=\emptyset$.
\end{theorem}
\begin{proof}
The  if ($\Leftarrow$) part follows from Theorem~\ref{thrm:shared_weak}. To see the (contrapositive of the) only if ($\Rightarrow$) direction, notice that under hypothesis $H_0$ (resp., $H_1$), the adversary may induce any channel $\convHull(\nullHyp)$ (resp., $\convHull(\altHyp)$) from the transmitter to the detector. Hence, when the intersection is non-empty, the adversary may induce the same channel under both hypotheses so that no transmission strategy (including an adaptive one) can distinguish between the hypotheses.    
\end{proof}

\section{Deterministic Transmitter} \label{sec:deterministic}
For $x \in \scriptX$, let $\convHull(\nullHyp_x)$ and $\convHull(\altHyp_x)$ be the convex hulls of the conditional distributions $W(.|x,s)$ and $\Wbar(.|x,\bar{s})$.
\begin{equation*}
\convHull(\nullHyp_x) \defineqq \inb{\sum \limits_{s \in \scriptS}P_{S}(s)W(.|x,s) : P_{S} \in \Delta_{\scriptS}},
\end{equation*}
$\convHull(\altHyp_x)$ is defined similarly with $\bar{S},\Wbar$ instead of $S,W$.
Define $\DOptDet$ to be
\begin{equation} \label{eqn:DoptDetDef}
\DOptDet := \max \limits_{x} \min_{\substack{U_x \in \convHull(\nullHyp_x) \\ \Ubar_x \in \convHull(\altHyp_x)}} D(U_x \| \Ubar_x),
\end{equation}
where $U_x(.) = U(.|x), \Ubar_x(.)=\Ubar(.|x)$.
\begin{theorem}[] \label{thrm:det_weak}
	Let $\nullHyp$ and $\altHyp$ be two sets of discrete memoryless channels which map $\mathcal{X}$ to $\mathcal{Y}$. For any $\eps \in (0,1)$, we have
	\begin{equation} \label{eqn:det_weak}
	\DOptDet \le \advExpDet(\nullHyp,\altHyp) \le \frac{\DOptDet}{1-\eps}.
	\end{equation}
    If 
	\begin{align}
	    \lim\limits_{t \rightarrow 0^{-}} \max \limits_{x} \inf_{\substack{U_x \in \convHull(\nullHyp_x) \\ \Ubar_x \in \convHull(\altHyp_x)}} \frac{\phi_{t}(U_x \| \Ubar_x)}{-t} = \max \limits_{x} \inf_{\substack{U_x \in \convHull(\nullHyp_x) \\ \Ubar_x \in \convHull(\altHyp_x)}} \lim\limits_{t \rightarrow 0^{-}} \frac{\phi_{t}(U_x \| \Ubar_x)}{-t}, \label{eqn:det_conv_assumption}
	\end{align}
	then
	\begin{equation}
	\advExpDet(\nullHyp,\altHyp) = \DOptDet.
	\end{equation}
    Furthermore,
    $\advExpDet(\nullHyp,\altHyp)>0 \iff \convHull(\nullHyp_x) \cap \convHull(\altHyp_x)=\emptyset$ for some $x$.
\end{theorem}

The proof (Appendix~\ref{app:det_weak_nofb}) is on similar lines as Theorems~\ref{thrm:shared_weak}, \ref{thrm:shared_strong}, \ref{thm:shared_characterization}. We also show that \eqref{eqn:det_weak} holds when both the transmitter and the adversary are adaptive (Appendix~\ref{app:det_weak_fb}). 

\section{Private Randomness}\label{sec:private}

We now consider the case where the transmitter may choose the channel input $X^n$ randomly, but the realization of $X^n$ is unavailable to the detector and the adversary. By the discussion in the proof of achievability in Theorem~\ref{thrm:shared_weak}, if the transmitter adopts an i.i.d. $P_X$ strategy, the best possible exponent (irrespective of whether the adversary is adaptive or not) is 
\begin{align*}
D_{\textup{pvt,iid}}=\sup_{P_{X}} \min_{\substack{Q_{Y} \in \mathcal{Q} \\ \bar{Q}_{Y} \in \bar{\mathcal{Q}}}} D(Q_{Y}\|\bar{Q}_{Y}),
\end{align*} 
where $\mathcal{Q}$ (resp. $\bar{\mathcal{Q}}$) is the set of (single-letter) channel output distributions that can be induced by the adversary  under hypothesis $H_0$ (resp. $H_1$) when the input is distributed as $P_X$, i.e.,
$\mathcal{Q} \defineqq \inb{\sum_{x,s} P_{S}(s)P_{X}(x)W(\cdot|x,s) : P_{S} \in \Delta_{\scriptS}}$.
It turns out that in general the optimal exponent $\bpvtexp$ could be strictly larger that $D_{\textup{pvt,iid}}$. In the following example, $\bpvtexp>0$ for all $\eps>0$ even though $D_{\textup{pvt,iid}}=0$.

\definecolor{xdxdff}{rgb}{0.49019607843137253,0.49019607843137253,1}
\definecolor{ududff}{rgb}{0.30196078431372547,0.30196078431372547,1}
\begin{figure}[!h]
	\resizebox{1\columnwidth}{!}{
	\begin{tikzpicture}[line cap=round,line join=round,>=triangle 45,x=1cm,y=1cm]
	\clip(-6.5,-5.06) rectangle (12.87,5.42);
	\draw [line width=1pt] (-6,2) -- (-2,2);
	\draw [line width=1pt] (-6,-2) -- (-2,-2);
	\draw [line width=1pt] (-6,-2) -- (-2,0);
	\draw [line width=1pt] (-6,2) -- (-2,0);
	\draw [line width=1pt] (2-0.5,2) -- (6-0.5,2);
	\draw [line width=1pt] (2-0.5,-2) -- (6-0.5,-2);
	\draw [line width=1pt] (2-0.5,2) -- (6-0.5,0);
	\draw [line width=1pt] (2-0.5,-2) -- (6-0.5,0);
	\draw [line width=1pt] (2-0.5,2) -- (6-0.5,-2);
	\draw [line width=1pt] (8,2) -- (12,2);
	\draw [line width=1pt] (8,-2) -- (12,-2);
	\draw [line width=1pt] (8,2) -- (12,0);
	\draw [line width=1pt] (8,-2) -- (12,0);
	\draw [line width=1pt] (8,-2) -- (12,2);

	\draw (-6.60,2.25) node[anchor=north west] {\LARGE $0$};
	\draw (-6.60,-1.75) node[anchor=north west] {\LARGE $1$};
	\draw (-1.83,2.25) node[anchor=north west] {\LARGE $0$};
	\draw (-1.85,0.2) node[anchor=north west] {\LARGE $e$};
	\draw (-1.87,-1.75) node[anchor=north west] {\LARGE $1$};
	
	\draw (-4.49,2.9) node[anchor=north west]  {\LARGE $1-p$};
	\draw (-4.47,-2.15) node[anchor=north west] {\LARGE $1-p$};
	\draw (-4.25,0.85) node[anchor=north west] {\LARGE $p$};
	\draw (-4.25,-0.40) node[anchor=north west] {\LARGE $p$};
	
	\draw (1.45-0.5,2.25) node[anchor=north west] {\LARGE $0$};
	\draw (1.45-0.5,-1.75) node[anchor=north west] {\LARGE $1$};
	\draw (6.15-0.5,2.25) node[anchor=north west] {\LARGE $0$};
	\draw (6.15-0.5,0.2) node[anchor=north west] {\LARGE $e$};
	\draw (6.15-0.5,-1.75) node[anchor=north west] {\LARGE $1$};
	
	\draw (2.2-0.5,2.9) node[anchor=north west] {\LARGE $(1-p)(1-r)$};
	\draw (4-0.5,1.60) node[anchor=north west] {\LARGE $p$};
	\draw (1.7-0.5,0.3) node[anchor=north west] {\LARGE $(1-p)r$};
	\draw (3.87-0.5,-1.10) node[anchor=north west] {\LARGE $p$};
	\draw (3.51-0.5,-2.15) node[anchor=north west] {\LARGE $1-p$};
	
	\draw (7.45,2.25) node[anchor=north west] {\LARGE $0$};
	\draw (7.45,-1.75) node[anchor=north west] {\LARGE $1$};
	\draw (12.15,2.25) node[anchor=north west] {\LARGE $0$};
	\draw (12.15,0.2) node[anchor=north west] {\LARGE $e$};
	\draw (12.11,-1.75) node[anchor=north west] {\LARGE $1$};
	
	\draw (9.53,2.90) node[anchor=north west] {\LARGE $1-p$};
	\draw (10,1.60) node[anchor=north west] {\LARGE $p$};
	\draw (7.75,0.45) node[anchor=north west] {\LARGE $(1-p)r$};
	\draw (9.89,-1.10) node[anchor=north west] {\LARGE $p$};
	\draw (8.1,-2.15) node[anchor=north west] {\LARGE $(1-p)(1-r)$};

	\draw (-5.0,-3.7) node[anchor=north west] {\huge $W(\cdot|\cdot)$};
	\draw (8.5,-3.7) node[anchor=north west] {\huge $\overline{W}(\cdot|\cdot,1)$};
	\draw (2,-3.7) node[anchor=north west] {\huge $\overline{W}(\cdot|\cdot,0)$};
	\draw [line width=1pt] (0-0.25,3)-- (0-0.25,-4);
	\begin{scriptsize}
	\draw [fill=ududff] (-6,2) circle (2.5pt);
	\draw [fill=ududff] (-2,2) circle (2.5pt);
	\draw [fill=ududff] (-6,-2) circle (2.5pt);
	\draw [fill=ududff] (-2,-2) circle (2.5pt);
	\draw [fill=xdxdff] (-2,0) circle (2.5pt);
	\draw [fill=ududff] (2-0.5,2) circle (2.5pt);
	\draw [fill=ududff] (6-0.5,2) circle (2.5pt);
	\draw [fill=ududff] (2-0.5,-2) circle (2.5pt);
	\draw [fill=ududff] (6-0.5,-2) circle (2.5pt);
	\draw [fill=xdxdff] (6-0.5,0) circle (2.5pt);
	\draw [fill=ududff] (8,2) circle (2.5pt);
	\draw [fill=ududff] (12,2) circle (2.5pt);
	\draw [fill=ududff] (8,-2) circle (2.5pt);
	\draw [fill=ududff] (12,-2) circle (2.5pt);
	\draw [fill=xdxdff] (12,0) circle (2.5pt);
	\end{scriptsize}
	\end{tikzpicture}
	}
	\caption{Example~\ref{ex:role_of_memory} considers two sets of channels $\nullHyp= \{ W(\cdot|\cdot) \}$ and $\altHyp=\{\Wbar(\cdot|\cdot,0),\Wbar(\cdot|\cdot,1)\}$ which cannot be distinguished using i.i.d. transmission schemes when the  transmitter is restricted to be privately randomized. However, a simple scheme with memory yields a positive Chernoff-Stein exponent.}
	\label{fig:fig1}
\end{figure}

\begin{example}[Figure~\ref{fig:fig1}] \label{ex:role_of_memory}
We define two sets of channels for the  alphabets $\scriptX = \{0,1\},\,\scriptY=\{0,1,e\},\, \scriptS =\{ 0 \} \text{ and } \scriptSbar = \{0,1\}$. The hypothesis $H_0:\nullHyp= \{ {W(\cdot|\cdot)} \}$ consists of a binary erasure channel with parameter $p<1$ (\bec{p}). The hypothesis $H_1$ consists of $\altHyp=\{\Wbar(\cdot|\cdot,0),\Wbar(\cdot|\cdot,1)\}$ where 
for any $\bar{s} \in \{0,1\}$, the channel  $\Wbar(\cdot|x,\bar{s})$ is defined as
\begin{align*}
\Wbar(e|x,\bar{s}) &= p \text{ and }\\
\Wbar(x|x,\bar{s}) &= \begin{cases}
    (1-p)(1-r) &\text{ if }\bar{s}=x,\\
    1-p &\text{ otherwise.}
\end{cases}
\end{align*} where $x\in \{0,1\}, r \in (0,1)$.
The channels $\Wbar(\cdot|\cdot,0)$ and $\Wbar(\cdot|\cdot,1)$ can be thought of as  modified \bec{p} channels where one of the symbols flips with probability $(1-p)r$ as shown in Figure~\ref{fig:fig1}.

Note that $\mathcal{Q}$ is a singleton, so there are no adversarial attacks. For any input distribution $P_{X}(0) = q$ and $P_X(1) = 1-q$, the induced output distribution is given by $P_Y(0)=\sum_xP_X(x)W(0|x) = q(1-p)$ and $P_Y(e)=\sum_xP_X(x)W(e|x) = p$.
On the other hand, under $H_1$, suppose the adversary sets $P_{\bar{S}}(0) = 1-P_{X}(0)$. Then, the induced channel output distribution is given by 
\begin{align*}
 P_Y(e)=\sum_{\bar{s}, x}&P_X(x)P_{\bar{S}}\Wbar(e|x,\bar{s})= p\text{ and }\\
P_Y(0)=\sum_{\bar{s}, x}&P_X(x)P_{\bar{S}}\Wbar(0|x,s)= q(1-q)\Wbar(0|0,0)+ q^2\Wbar(0|0,1) + (1-q)^2\Wbar(0|1,0) + (1-q)q\Wbar(0|1,1)\\
& = q(1-q)(1-p)(1-r)+ q^2(1-p)+ 0 + (1-q)q(1-p)r\\
& = q(1-p).
 \end{align*} 
This is the  same as the one under $H_0$. Hence, $\mathcal{Q} \subset \bar{\mathcal{Q}}$ and therefore $D_{\textup{pvt,iid}} = 0$.

	Now to see that $\bpvtexp>0$, consider a transmission scheme with 2-step memory: $n/2$ i.i.d. pairs are sent where each pair is distributed as $P_{X_1,X_2}(0,0) = P_{X_1,X_2}(1,1) = 0.5$. The effective channel is now a random map from $\scriptX^2$ to $\scriptY^2$. The new state space for the (non-adaptive) adversary under $H_0$ is $\scriptS^2$ (which is still a singleton), and $\scriptSbar^2$ under $H_1$. Let $\mathcal{Q}_2$ (resp. $\bar{\mathcal{Q}}_2$) be the set of (two-letter) channel output distributions that can be induced by the adversary when the input is distributed according to $P_{X_1,X_2}$ under $H_0$ (resp. $H_1$). Since $\mathcal{Q}_2$ is a singleton, let the member be denoted by $Q_{Y_1,Y_2}$. If we show that $Q_{Y_1,Y_2} \notin \bar{\mathcal{Q}}_2$, we may conclude that $\bpvtexp > 0$. Assume for contradiction that this is not the case, i.e., suppose there exists $P_{\bar{S}_1,\bar{S}_2}$ such that the resulting $\bar{Q}_{Y_1,Y_2}$ is the same as $Q_{Y_1,Y_2}$. Since the marginals also have to be equal, we have $\bar{Q}_{Y_1}=Q_{Y_1}=(\frac{1-p}{2},p,\frac{1-p}{2})$. Let $P_{\bar{S}_1}(0)=t$. Then, $\bar{Q}_{Y_1}(0)$ is given by
    \begin{align*}
        &\sum_{x,\bar{s}_1} P_X(x)P_{\bar{S}_1}(\bar{s}_1)\Wbar(0|x,\bar{s}_1) \\
        &= \frac{1}{2} t (1-p)(1-r) + \frac{1}{2} (1-t) (1-p) + \frac{1}{2} (1-t) (1-p)r \\
        &=\frac{1-p}{2} (1+r-2tr).
    \end{align*}
    This forces $t=0.5$, i.e. $P_{\bar{S}_1}$ has to be uniform. Now, observe that $Q_{Y_1,Y_2}(0,1)=0$ while, irrespective of $P_{\bar{S}_2|\bar{S}_1}$, we have $\bar{Q}_{Y_1,Y_2}(0,1) > 0$ since $r>0$. This is a contradiction and hence $Q_{Y_1,Y_2} \notin \bar{\mathcal{Q}}_2$. Therefore, $\bpvtexp > 0$ for all $\eps>0$ by Theorem~\ref{thrm:bhlp1}.
	
	The above argument does not account for an adaptive adversary. In Appendix~\ref{app:example1} we show that even with an adaptive adversary the above transmission scheme leads to a positive exponent.
\end{example}
\begin{remark}\label{rem:det-rand-separation}
Observe that $\DOptDet\leq D_{\textup{pvt,iid}}$. This is a consequence of the fact that for $P_X$ such that $P_X(x)=1$ for some $x\in\cX$, the corresponding $\cQ$ and $\bar{\cQ}$ are $\convHull(\nullHyp_x)$ and $\convHull(\altHyp_x)$ respectively. In the above example, we conclude that $0=D_{\textup{pvt,iid}}<\bpvtexp$ for all $\eps > 0$; therefore, we have $\advExpDet\inp{\nullHyp,\altHyp}<\bpvtexp$. 
\end{remark}

\begin{remark} Example~\ref{ex:role_of_memory} demonstrates that, in the setting of encoders with private randomness, the error exponent can be strictly improved by drawing channel inputs i.i.d. as blocks of two symbols (instead of just one symbol at a time). It is conceivable that, in general, the optimal error exponent could only be achieved asymptotically  through a sequence of schemes that rely on drawing channel inputs as blocks of increasing length. Towards this, ~\cite{bakshi2025} gives an example of a channel, over which schemes involving drawing inputs as blocks of length $3$ strictly improve upon schemes that involve drawing inputs as blocks of length $2$.
\end{remark}
In the rest of this section, we give an achievable lower bound on the error exponent $\bpvtexp$ and characterize the pairs $\inp{\nullHyp,\altHyp}$ for which it is positive\footnote{This characterization is implicit in~\cite[Corollary 1]{chaudhuri2021compoundArxiv}. Note that the ``deterministic coding'' transmitter there has access to the message that serves as a source of private randomness for the testing problem.}. 
If $\convHull(\nullHyp)\cap\convHull(\altHyp)\neq \emptyset$, then $\bpvtexp = 0$ (by Theorem~\ref{thm:shared_characterization}). This follows from the fact that the adversary can choose $S^n$ and $\bar{S}^n$ i.i.d. so that a channel in the intersection may be induced, which renders the hypotheses indistinguishable irrespective of the transmission scheme. It turns out that when the transmitter only has private randomness, a more carefully chosen adversary strategy which now depends on the transmission scheme may render $\bpvtexp = 0$ for a larger class of $\inp{\nullHyp,\altHyp}$ pairs.
\begin{definition}[\hspace{-0.02cm}{\cite[eq. $(2)$]{chaudhuri2021compound}}]\label{defn:trans-sym}
The pair $\inp{\nullHyp,\altHyp}$ is {\em trans-symmetrizable} if there exist conditional distributions 
$P_{S|X},P_{\bar{S}|X}$ such that, for every $x, \tilde{x}\in\cX$ and $y\in \cY$,
\begin{align}\label{eq:trans_sym}
\sum_{s \in \cS}P_{S|X}(s|x)W(y|\tilde{x},s)=\sum_{\bar{s} \in \bar{\cS}}P_{\bar{S}|X}(\bar{s}|\tilde{x})\Wbar(y|x,\bar{s}).
\end{align}
\end{definition} 

Trans-symmetrizability was shown to be a unique condition, not a consequence of symmetrizability of either of the AVCs. Non-trans-symmetrizability and disjointness of the convex hulls of channel sets was shown to be necessary and sufficient for the detection of the hypothesis with vanishing error probabilities~\cite[Corollary 1]{chaudhuri2021compoundArxiv}. In Theorem 3 below, we show that the same condition is also necessary and sufficient for achieving a non-zero error-exponent. We also provide a lower bound on the error exponent when it is positive.

Consider a trans-symmetrizable pair $\inp{\nullHyp,\altHyp}$ and a (non-adaptive\footnote{This discussion can be modified to handle an adaptive transmission scheme if the adversary is also adaptive.}) transmission scheme $\hat{P}$. We will demonstrate (non-adaptive) adversary strategies under which the detector is unable to distinguish between the hypotheses. Under hypothesis $H_1$, the adversary, independent of the transmitter, samples $\tilde{X}^n$ according to $\hat{P}$ and passes it through the (memoryless) channel $P_{\bar{S}|X}$ of Definition~\ref{defn:trans-sym} to produce $\bar{S}^n$. This induces the following distribution on the channel output vector:
\begin{align*}
&\sum_{x^n,{\bar{s}}^n}\hat{P}(x^n)\left[\sum_{\tilde{x}^n}\hat{P}(\tilde{x}^n)\prod_{i=1}^{n}\inp{P_{\bar{S}|X}(\bar{s}_i|\tilde{x}_{i})}\right]\Wbar^n(y^n|x^n,{\bar{s}}^n)\nonumber\\
&=
\sum_{x^n, \tilde{x}^n}\hat{P}(x^n)\hat{P}(\tilde{x}^n)\prod_{i=1}^{n}\left[\sum_{\bar{s}_i\in \bar{\cS}}P_{\bar{S}|X}(\bar{s}_i|\tilde{x}_{i})\Wbar(y_i|x_i,\bar{s}_i)\right]\nonumber\\
&\stackrel{(a)}{=}
\sum_{\tilde{x}^n, x^n}\hat{P}(\tilde{x}^n)\hat{P}(x^n)\prod_{i=1}^{n}\left[\sum_{{s}_i\in {\cS}}P_{{S}|X}(s_i|x_i)W(y_i|\tilde{x}_{i},s_i)\right]\nonumber\\
&=
\sum_{\tilde{x}^n,{s}^n}\hat{P}(\tilde{x}^n)\left[\sum_{x^n}\hat{P}(x^n)\prod_{i=1}^{n}\inp{P_{{S}|X}(s_i|x_i)}\right]W^n(y^n|\tilde{x}^n,{s^n})
\end{align*}
where $(a)$ follows from \eqref{eq:trans_sym}. This is identical to the channel output distribution under hypothesis $H_0$ if the adversary samples from $\hat{P}$ (independent of the transmitter) and passes through the channel $P_{{S}|X}$ of Definition~\ref{defn:trans-sym} to produce its ${S}^n$. Thus, $\bpvtexp = 0$ if $\inp{\nullHyp,\altHyp}$ is trans-symmetrizable. 
The example below 
establishes a separation between shared and private randomness.
\begin{example}[{\hspace{-0.02cm}\cite[Example~1]{chaudhuri2021compound}}]\label{ex:pvt-shared-separation}
Let  $\cX = \cS = \bar{\cS} = \{0,1\}$ and $\cY = \{0,1\}^2$. Suppose $W$ deterministically outputs $Y=(X,S)$ while $\Wbar$ outputs $Y=(\bar{S},X)$. Clearly, $\convHull(\nullHyp)\cap\convHull(\altHyp)=\emptyset$. Hence, by Theorem~\ref{thm:shared_characterization}, $\advExpShared>0$. However, $\inp{\nullHyp,\altHyp}$ is trans-symmetrizable since $P_{S|X}(x|x) = P_{\bar{S}|X}(x|x)=1$ for all $x\in \cX$ satisfies \eqref{eq:trans_sym}. Hence $\bpvtexp = 0$.
\end{example}

Note that  if $\convHull(\nullHyp)\cap\convHull(\altHyp)=\emptyset$, there exists a constant $\zeta_1>0$ such that for every $P_{\bar{S}}$ on $\bar{\cW}$ and $P_{S}$ on $\cW$, 
 \begin{align}
 &\max_{x, y}\inl{\sum_{\bar{s}}P_{\bar{S}}(\bar{s})\bar{W}(y|x, \bar{s})- \sum_{s}P_{S}(s)W(y|x, s)}>\zeta_1.\label{eq:gap_empty_intersection}
\end{align}

Also, if $\inp{\nullHyp,\altHyp}$ is not trans-symmetrizable, there exists $\zeta_2>0$ such that for every $P_{S|X}(s|x'),\, s\in \cS, x'\in \cX$ and $P_{\bar{S}|X}(\bar{s}|x), \bar{s}\in \bar{\cS}, x\in \cX$ 

\begin{align} \label{eq:gap_to_sym}
&\max_{x, x', y}\inl{\sum_{s \in \scriptS}P_{S|X}(s|x')W(y|x,s)-\sum_{\bar{s} \in \scriptSbar}P_{\bar{S}|X}(\bar{s}|x)W(y|x',\bar{s})} > \zeta_2. 
\end{align}

Our lower bound on $\bpvtexp$ is in terms $\zeta_1$ and $\zeta_2$ which quantitatively measure respectively how far the pair $\inp{\nullHyp,\altHyp}$ is from having a non-empty intersection of their convex hulls and how far it is from being trans-symmetrizable; Lemma~\ref{app_lemma:disambiguity} and its proof in Appendix~\ref{app:pvt_rand} make this connection concrete.

Our main theorem for this section is the following:
\begin{theorem}\label{thm:pvt_rand}
\begin{align*}
&\bpvtexp = 0\text{ if } \inp{\nullHyp,\altHyp} \text{ is trans-symmetrizable or }\\
 &\qquad\qquad\qquad\quad\;\convHull(\nullHyp)\cap\convHull(\altHyp)\neq \emptyset.\\
\intertext{Otherwise,}
&\bpvtexp \geq  \max \inb{\min\inb{\frac{\zeta_1^2}{5|\cX|^2}, \frac{\zeta_2^2}{11|\cX|^4}}, \advExpDet(\nullHyp,\altHyp)} .
\end{align*}
\end{theorem}




Since $\zeta_1>0$ if $\convHull(\nullHyp)\cap\convHull(\altHyp)=\emptyset$ and $\zeta_2>0$ if $\inp{\nullHyp,\altHyp}$ is not trans-symmetrizable, we have the following characterization of pairs $\inp{\nullHyp,\altHyp}$ for which the Chernoff-Stein exponent $\bpvtexp$ is positive.
\begin{corollary}
$\bpvtexp >0$ if and only if $\inp{\nullHyp,\altHyp}$ is not trans-symmetrizable and $\convHull(\nullHyp)\cap\convHull(\altHyp)= \emptyset$.
\end{corollary}
This recovers \cite[Corollary 1]{chaudhuri2021compoundArxiv} which gave the same characterization for $\inp{\nullHyp,\altHyp}$ which allow hypothesis testing with vanishing probability of error when the transmitter has private randomness (in the form of a random message there). Our proof (in Appendix~\ref{app:pvt_rand}) of the lower bound to $\bpvtexp$ in Theorem~\ref{thm:pvt_rand}, which is inspired by \cite{chaudhuri2021compoundArxiv}, entails significant careful modifications to the detector and the probability of error analysis there.

\section{On the Role of Adaptivity}\label{sec:adaptive}
\subsubsection{With shared randomness} Our results hold even if the transmitter and/or adversary is adaptive. We proved the achievability part of Theorem~\ref{thrm:shared_weak} assuming that the adversary is adaptive and the converse assuming the transmitter is adaptive.

\subsubsection{Deterministic schemes} Here the optimal exponent remains unchanged even if the adversary is adaptive. This is also the case if both the adversary and the transmitter are adaptive. These follow from our achievability proof which is shown assuming an adaptive adversary and the converse which is shown when (a) both the transmitter and adversary are non-adaptive and (b) when both are adaptive (see Appendix~\ref{app:det_weak_fb}). It is also easy to see that, in general, if the transmitter is adaptive and the adversary is not, the exponent could be improved. The transmitter and detector may extract some randomness unknown to the adversary from the channel output feedback of, say, the first half of the block, and use this to implement a scheme with shared randomness during the second half. Since there are channels for which deterministic exponent is zero while the exponent under shared randomness is positive (for instance, see Example~\ref{ex:pvt-shared-separation}), these (possibly augmented by an independent random channel output component which provide additional shared randomness) serve as examples where such an improvement is feasible.

\subsubsection{With private randomness} If the adversary is non-adaptive and the transmitter is adaptive, improved exponents are possible along the lines of the above discussion, i.e. feedback from the detector to the transmitter can be used to simulate shared randomness. There are channels where the exponent with shared randomness is positive, while that with private randomness is zero (specifically, trans-symmetrizable but with $\convHull(\nullHyp) \cap \convHull(\altHyp) = \emptyset$; see Example~\ref{ex:pvt-shared-separation}). We also showed that memoryless schemes may be strictly sub-optimal even if the adversary is adaptive (Appendix~\ref{app:example1}). Also, the impossibility result in Theorem~\ref{thm:pvt_rand} can be shown when both the transmitter and adversary are adaptive. 

%
%

\section{Sequential setting} \label{sec:seq_shared}
In this section we study sequential versions of the problems covered in the previous sections. We show that in the sequential setting we can simultaneously achieve the two Chernoff-Stein exponents in each of the three settings: (i) shared randomness, (ii) deterministic and (ii) private randomness. We describe the problem in the sequential setting for the shared randomness case. The description for the other two cases are similar.

The test now comprises of a transmitter strategy, stopping time and a decision rule. A sequential test $\phi$ is defined by the tuple $(\hat{P},\tau, Z)$, where $\hat{P}=P_{X_1}P_{X_2|X_1}P_{X_3|X_1,X_2}\cdots$ is the transmitter strategy, $\tau$ is a stopping time of the filtration $\mathcal{F}_0 \subseteq \mathcal{F}_1 \cdots \subseteq \mathcal{F}_t \cdots \subseteq \mathcal{F}$ where $\mathcal{F}_t :=\ \sigma\{X_1,Y_1,\ldots,X_t,Y_t\}$, and $Z:\mathcal{F}_{\tau}\to\{0,1\}$ is a $\mathcal{F}_{\tau}$-measurable function that specifies the decision rule applied by the detector. Let $\mathcal{H}$ denote the set of all stopped sequences. Let $A$ be the acceptance region for $H_0$, i.e., stopped sequences which map $Z$ to $0$. Let $\hat{P}_S=P_{S_1}P_{S_2|S_1}P_{S_3|S_1,S_2}\cdots$ be the adversary strategy under $H_0$ and $\hat{P}_{\bar{S}}=P_{\bar{S}_1}P_{\bar{S}_2|\bar{S}_1}P_{\bar{S}_3|\bar{S}_1,\bar{S}_2} \cdots$ be the adversary strategy under $H_1$. For a given transmitter strategy $\hat{P}$ and a pair of adversary strategies $\hat{P}_S$ and $\hat{P}_{\bar{S}}$, let $\Qshared$ and $\Qbarshared$ be the measures on $X_1,Y_1,X_2,Y_2,\ldots$ under $H_0$ and $H_1$ respectively. Thus, the joint distribution of $X_1, Y_1, X_2, Y_2 \cdots X_t,Y_t$ under $H_0$ is given by
\begin{align} \label{eqn:q_shared_seq}
\Qshared(x^t,y^t) = \sum \limits_{s^t} \prod_{i=1}^{t} P_{X_i|X^{i-1}}(x_i|x^{i-1})P_{S_i|S^{i-1}}(s_i|s^{i-1})W(y_i|x_i,s_i).
\end{align}
A similar expression is can be written for $\Qbarshared(x^t,y^t)$ under $H_1$. The type-I error is given by
\begin{equation*}
\alpha(\phi, \hat{P}_S) = \Qshared(A^c).
\end{equation*}
The type-II error is given by
\begin{equation*}
\beta(\phi, \hat{P}_{\bar{S}}) \defineqq \Qbarshared(A).
\end{equation*}
If the test is randomized, then we can take an expectation over the random choice of $A$. A pair of exponents $(E_0,E_1)$ is said to be \emph{achievable in the sequential sense}, if there exists a sequence of tests $(\phi_n=(\hat{P}_n,\tau_n, Z_n))_{n \in \mathbb{N}}$ such that
\begin{align*}
\liminf_{n \rightarrow \infty} -\frac{1}{n} \log \sup_{\hat{P}_S} \alpha(\phi_n, \hat{P}_S) &\ge E_0  \\
\liminf_{n \rightarrow \infty} -\frac{1}{n} \log \sup_{\hat{P}_{\bar{S}}} \beta(\phi_n, \hat{P}_{\bar{S}}) &\ge E_1,
\end{align*}
and $\sup_{\hat{P}_S} \mathrm{E}[\tau_n] \le n$, $\sup_{\hat{P}_{\bar{S}}} \mathrm{E}[\tau_n] \le n$ for $n > n_0$ for some large enough $n_0$. We define $\mathcal{E}_{\textup{sh}}^{\textup{seq}}(\nullHyp,\altHyp)$ to be the set of achievable pair of exponents for the shared randomness case. 

For the deterministic case, $\hat{P}$ is a point mass on a fixed sequence $(x_1,x_2, \ldots)$. For the deterministic and private randomness case, $\tau$ is a stopping time with respect to the filtration $\mathcal{F}_0 \subseteq \mathcal{F}_1 \cdots \subseteq \mathcal{F}_t \cdots \subseteq \mathcal{F}$ where $\mathcal{F}_t := \sigma\{Y_1,\ldots,Y_t\}$, $Z:\mathcal{F}_{\tau}\to\{0,1\}$ is a $\mathcal{F}_{\tau}$-measurable decision rule, and the acceptance region $A$ is the subset of stopped sequences for which the detector accepts $H_0$. $Q_{\textup{priv}}(y^t)$ can be obtained by marginalizing $\Qshared(x^t,y^t)$ over $x^t$. $Q_{\textup{det}}(y^t)$ can be obtained from $\Qshared(x^t,y^t)$ by replacing the input distribution with a point mass on the fixed sequence. Let $\mathcal{E}_{\textup{det}}^{\textup{seq}}(\nullHyp,\altHyp)$ and $\mathcal{E}_{\textup{priv}}^{\textup{seq}}(\nullHyp,\altHyp)$ be the set of achievable pair of exponents for the deterministic and private randomness cases respectively.

\paragraph*{Fixed-Length to Sequential Tests} We first outline  the general form of the sequential tests employed in the schemes given in this section. Recall that a fixed length scheme is a pair $\phi_n =(\hat{P}_n,n,Z_n)$ (i.e., a sequential test with $\tau_n=n$). Let $\mathbb{P}$ (resp. $\mathbb{Q}$) be the distribution induced on the observations under $H_0$ (resp. $H_1$). Thus, the type-I and type-II errors are given by
\begin{align*}
    \bar{\alpha}(\phi_n) &= \inf_{\phi_n} \sup_{\hat{P}_S} \mathbb{P}(Z_n=1) \\
    \bar{\beta}(\phi_n) &= \inf_{\phi_n} \sup_{\hat{P}_{\bar{S}}} \mathbb{Q}(Z_n=0).
\end{align*}
Concretely, assume the existence of three fixed length schemes. The first scheme is such that both types of errors decay to zero with increasing blocklength. The second (resp. third) scheme is such that it achieves an exponent for type-I error (resp. type-II error) while driving the type-II error (resp. type-I error) to zero with increasing blocklength. Observe that the second (or third) scheme satisfies the requirements of the first scheme, but for the sake of exposition, we keep them separate. Let $\theta,\gamma \in (0,1)$ be two parameters which we will set later. The sequential test proceeds in rounds. Each round is of length $n'=(1-\theta)n$ and consists of two phases. The first phase (or trial phase) is of length $\gamma n'$. In this phase, we need a scheme which can make a (tentative) decision such that both the types of errors decay to zero as the block-length ($\gamma n'$) goes to infinity. The second phase (or confirmation phase) is of length $(1-\gamma)n'$. In this phase, depending on the trial phase decision we use a scheme to achieve the corresponding (fixed length) Chernoff-Stein exponent. For example, if the trial phase decision was $\nullHyp$ we use the scheme which achieves the exponent in Theorem~\ref{thrm:shared_weak}. If the decisions in the confirmation and trial phases match, we stop. Else, we go to the next round. The parameters $\theta, \gamma$ are chosen so that the expected stopping time is less than or equal to $n$.

The following lemma shows how the three schemes can be combined to simultaneously achieve exponents for both types of errors in the sequential case. 
\begin{lemma} \label{lemma:seq_test}
    Let $\{T_n\}$ be a sequence of fixed length schemes such that $\bar{\alpha}(T_n),\bar{\beta}(T_n) \rightarrow 0$ as $n \rightarrow \infty$. Let $\{C_n^0\}$ be a sequence of fixed length schemes such that $\bar{\beta}(C_n^0) \rightarrow 0$ as $n \rightarrow \infty$ and 
    \begin{equation*}
    \liminf_{n \rightarrow \infty} -\frac{1}{n} \log \bar{\alpha}(C_n^0) = E_0.    
    \end{equation*}
    Let $\{C_n^1\}$ be a sequence of fixed length schemes such that $\bar{\alpha}(C_n^1) \rightarrow 0$ as $n \rightarrow \infty$ and 
    \begin{equation*}
     \liminf_{n \rightarrow \infty} -\frac{1}{n} \log \bar{\beta}(C_n^1) = E_1.   
    \end{equation*}
    Then the point $(E_0,E_1)$ is \emph{achievable in the sequential sense} using tests made up of repeated use of fixed length test sequences $\{T_n\}$, $\{C_n^0\}$ and $\{C_n^1\}$.
\end{lemma}
\begin{proof}
    We construct a sequence of \emph{sequential} tests $\{\phi_n\}$ as follows. Let $\theta, \gamma > 0$ be two parameters whose values wil be specified later. The test $\phi_n$ consists of rounds of length $n'=(1-\theta)n$. At the beginning of each round $r$, we use the scheme $T_{\gamma n'}$ on the first $\gamma n'$ symbols and output $\tilde{Z}_r$ (known as tentative decision). If $\tilde{Z}_r = 0$ (respectively $\tilde{Z}_r = 1$), we then use scheme $C_n^1$ (resp. $C_n^0$) on the remaining $(1-\gamma) n'$ symbols and output $C_r$ (known as confirmation decision). If $(C_r = \tilde{Z}_r)$, we stop and declare $C_r$ to be our decision. Else, we proceed to the next round and repeat. We first show that $\mathrm{E}_{\mathbb{P}}[\tau_n] \le n$, $\mathrm{E}_{\mathbb{Q}}[\tau_n] \le n$ for some large enough $n$. We show the former, the latter follows by a symmetric argument. Let $R$ be the random variable denoting the number of rounds until the confirmation decision matches the tentative decision. Thus, $\tau_n = n'R$. Thus, to get an upper bound on the expected stopping time, we need an upper bound on the expected number of rounds. Observe that $R$ is a geometric random variable. Since we want an upper bound on its expected value, it suffices to upper bound the failure probability. We go to the next round in the event $\{C_r \neq \tilde{Z}_r\}$. We have
        \begin{align*}
            \mathbb{P}(C_r \neq \tilde{Z}_r) &= \mathbb{P}(\tilde{Z}_r=0,C_r=1) + \mathbb{P}(\tilde{Z}_r=1,C_r=0) \\
            &\le \mathbb{P}(C_r=1|\tilde{Z}_r=0) + \mathbb{P}(\tilde{Z}_r=1)
        \end{align*}
        Observe that 
        \begin{equation*}
            \mathbb{P}(\tilde{Z}_r=1) = \bar{\alpha}(T_{\gamma n'}).
        \end{equation*}
        Also, we know that
        \begin{align*}
            \mathbb{P}(C_r=1|\tilde{Z}_r=0) = \bar{\alpha}(C_{(1-\gamma)n'}^1)
        \end{align*}
        Plugging in $n'=(1-\theta)n$, we get
        \begin{align*}
           \mathbb{P}(C_r \neq \tilde{Z}_r) &\le \bar{\alpha}(T_{\gamma n'}) + \bar{\alpha}(C_{(1-\gamma)n'}^1).
        \end{align*}
        The expected stopping time can now be bounded as
        \begin{align*}
            \mathbb{E}[\tau_n] &= n' \mathbb{E}[R] \\
            &\le (1-\theta)n\frac{1}{1-(\bar{\alpha}(T_{\gamma n'}) + \bar{\alpha}(C_{(1-\gamma)n'}^1))} \\
            &= (1-\theta)n\frac{1}{1-(\bar{\alpha}(T_{\gamma (1-\theta)n}) + \bar{\alpha}(C_{(1-\gamma)(1-\theta)n}^1))}.
        \end{align*}
        Since $\bar{\alpha}(T_{\gamma (1-\theta)n}), \bar{\alpha}(C_{(1-\gamma)(1-\theta)n}^1)) \rightarrow 0$ as $n \rightarrow \infty$, for large enough $n$, we have
        \begin{equation*}
            \theta > \bar{\alpha}(T_{\gamma (1-\theta)n}) + \bar{\alpha}(C_{(1-\gamma)(1-\theta)n}^1)).
        \end{equation*}
        Thus, we have $\mathbb{E}_{\mathbb{P}}[\tau_n] \le n$ for large enough $n$.
        We now analyse the error exponents. Observe that the error under $H_0$ happens when $\tilde{Z}_R=C_R=1$, i.e. our tentative decision in the trial phase is wrong and we confirm it in the confirmation phase. Thus, $\bar{\alpha}(\phi_n) = \mathbb{P}(\tilde{Z}_R=1, C_R=1)$. For any fixed $r$, we have
        \begin{align*}
            \mathbb{P}(\tilde{Z}_r=1, C_r=1) &\le \mathbb{P}(C_r=1|\tilde{Z}_r=1) \\
            &\overset{(a)}{=} \bar{\alpha}(C_{(1-\gamma)n'}^0) \\
            &\overset{(b)}{\le} 2^{-(1-\gamma)E_0 n'} \\
            &= 2^{-(1-\gamma)(1-\theta) E_0 n}.
        \end{align*}
        Here, $(a)$ follows from the construction of our test and $(b)$ follows from the property of the scheme $C_n^0$. Since this holds for all $r$, we have
        \begin{equation*}
            \bar{\alpha}(\phi_n) \le 2^{-(1-\gamma)(1-\theta) E_0 n}.
        \end{equation*}
        Now, for any $\eta > 0$, we can choose $\gamma, \theta$ such that
        \begin{equation*}
            \liminf_{n \rightarrow \infty} -\frac{1}{n} \log \bar{\alpha}(\phi_n) \ge E_0-\eta.
        \end{equation*}
        The error analysis under $H_1$ can be done similarly. This completes the proof.
\end{proof}

In the upcoming sections, we will elaborate on the precise forms of the trial and confirmation schemes for each of the three settings.

\subsection{Shared randomness}
Let $\DOptShBar$ and $\DOptSh$ be the Chernoff-Stein exponents for the type-I and type-II errors respectively. Recall that
\begin{equation*}
    \DOptSh := \sup \limits_{P_X} \min_{\substack{U \in \convHull(\nullHyp) \\ \Ubar \in \convHull(\altHyp)}} D(U\|\Ubar|P_{X}),  
\end{equation*}
and
\begin{equation*}
    \DOptShBar := \sup \limits_{P_X} \min_{\substack{U \in \convHull(\nullHyp) \\ \Ubar \in \convHull(\altHyp)}} D(\Ubar\|U|P_{X}).
\end{equation*}

\begin{theorem}[] \label{thrm:shared_weak_seq}
	Let $\nullHyp$ and $\altHyp$ be two sets of discrete memoryless channels which map $\mathcal{X}$ to $\mathcal{Y}$. The set of achievable pairs of exponents is given by
	\begin{equation} \label{eqn:shared_weak_seq}
	\mathcal{E}_{\textup{sh}}^{\textup{seq}}(\nullHyp,\altHyp) = \inb{ (E_0,E_1):E_0 \le \DOptShBar, E_1 \le \DOptSh}.
	\end{equation}
\end{theorem}
\begin{proof}
    {\em Achievability:}
        We will construct a sequential test $\phi_n$ by repeated use of fixed length tests $T_n, C_n^0,C_n^1$ on the lines of Lemma~\ref{lemma:seq_test}. First fix distribution $\tilde{P}_X$ such that $\tilde{P}_X U \neq \tilde{P}_X \overline{U}$ for all $U \in \convHull(\nullHyp), \Ubar \in \convHull(\altHyp)$. If such a distribution doesn't exist, then the exponents will be zero.
        We first describe $C_n^1$. It is exactly the fixed length scheme that achieves the Chernoff-Stein exponent (for the type-II error) in Theorem~\ref{thrm:shared_weak}. $C_n^0$ is the fixed length scheme that achieves the Chernoff-Stein exponent (for the type-I error). The above schemes satisfy the conditions required Lemma~\ref{lemma:seq_test} with $E_0=\DOptShBar,E_1=\DOptSh$ respectively (refer Theorems~\ref{thrm:shared_weak}, \ref{thrm:bhlp1}). One of these tests can also be used as $T_n$ since we just need $\bar{\alpha}(T_n),\bar{\beta}(T_n) \rightarrow 0$ as $n \rightarrow \infty$. Invoking Lemma~\ref{lemma:seq_test} completes the proof of achievability.

    {\em Converse:}
    Assume that there is a sequence of sequential tests $(\phi_n)_{n \in \mathbb{N}}$ that achieves the pair $(E_0,E_1)$ such that $E_0>0,E_1>0$. Fix the following attack strategy. Under $H_0$ the adversary chooses $P_S$ i.i.d such that $\sum_{s}P_S(s)W(.|.,s)=U'$ and under $H_1$ it chooses $P_{\bar{S}}$ such that $\sum_{\bar{s}}P_{\bar{S}}(\bar{s})\Wbar(.|.,\bar{s})=\Ubar'$. The choice of $U',\Ubar'$ will be specified later. We give the converse argument under the assumption that the transmitter has feedback. Thus, $\Qshared$ for a $t$ length sequence can be written as
    \begin{equation*}
        \Qshared(x^t,y^t) = \prod_{i=1}^{t} P_{X_i|X^{i-1},Y^{i-1}}(x_i|x^{i-1},y^{i-1})U'(y_i|x_i).
    \end{equation*}
    $\Qbarshared$ can be written similarly replacing $U'$ with $\Ubar'$. Let $\Qshared|_{\mathcal{F}_{\tau_n}}$ and $\Qbarshared|_{\mathcal{F}_{\tau_n}}$ be the measures restricted to $\mathcal{F}_{\tau_n}$, i.e. the set of stopped sequences. By data processing inequality, we have
    \begin{align} \label{eqn:seq_shared_conv_dp}
       D(\text{Bern}(\alpha(\phi_n,\hat{P}_S))\|\text{Bern}(1-\beta(\phi_n, \hat{P}_{\bar{S}})))
       \le D(\Qshared|_{\mathcal{F}_{\tau_n}}\|\Qbarshared|_{\mathcal{F}_{\tau_n}}).
    \end{align}
    The R.H.S. in \eqref{eqn:seq_shared_conv_dp} can be decomposed as follows
    \begin{align*}
    D(\Qshared|_{\mathcal{F}_{\tau_n}}\|\Qbarshared|_{\mathcal{F}_{\tau_n}}) &\overset{(a)}{=} \mathbb{E}_{\Qshared|_{\mathcal{F}_{\tau_n}}}\insq{\log \prod_{i=1}^{\tau_n} \frac{U'(Y_i|X_i)}{\Ubar'(Y_i|X_i)}} \\
	&= \mathbb{E}_{\Qshared|_{\mathcal{F}_{\tau_n}}}\insq{\sum_{i=1}^{\tau_n} \log \frac{U'(Y_i|X_i)}{\Ubar'(Y_i|X_i)}}.
	\end{align*}
    The simplified form $(a)$ is because the $P_{X_i|X^{i-1},Y^{i-1}}(X_i|X^{i-1},Y^{i-1})$ terms cancel out. For brevity, we will drop $\Qshared|_{\mathcal{F}_{\tau_n}}$ from the subscript of the expectation. Let $S_{\tau_n}$ be the log-likelihood ratio.
    \begin{equation*}
        S_{\tau_n} := \sum_{i=1}^{\tau_n} \log \frac{U'(Y_i|X_i)}{\Ubar'(Y_i|X_i)}
    \end{equation*}
    Let $(V_{x,t})_t$ be the sequence of i.i.d. samples obtained when input symbol $x$ is chosen. 
    Let $N_x$ be a random variable denoting the number of times the input symbol $x$ was chosen. Observe that $\tau_n = \sum_{x \in \mathcal{X}} N_x$. Then, $S_{\tau_n}$ can be rewritten as
    \begin{align*}
        S_{\tau_n} = \sum_{x \in \sX} \sum_{t=1}^{N_x} \log \frac{U'(V_{x,t}|x)}{\Ubar'(V_{x,t}|x)}
    \end{align*}
    By applying Wald's lemma to $S_{\tau_n}$ (see proof of Lemma 1, \cite{kaufmann2016complexity}), we get
    \begin{equation*}
        \mathbb{E}[S_{\tau_n}] =  \sum_{x \in \sX} \mathbb{E}[N_x] D(W'(.|x) \| \Wbar'(.|x))
    \end{equation*}
    Thus, the R.H.S. in \eqref{eqn:seq_shared_conv_dp} can be written as
    \begin{align*}
        D(\Qshared|_{\mathcal{F}_{\tau_n}}\|\Qbarshared|_{\mathcal{F}_{\tau_n}}) &= \sum_{x \in \sX} \mathbb{E}[N_x] D(U'(.|x) \| \Ubar'(.|x)) \\
        &= \mathbb{E}[\tau_n] \sum_{x \in \sX} \frac{\mathbb{E}[N_x]}{\mathbb{E}[\tau_n]} D(U'(.|x) \| \Ubar'(.|x)) \\
        &\overset{(a)}{\le} n \sum_{x \in \sX} \frac{\mathbb{E}[N_x]}{\mathbb{E}[\tau_n]} D(U'(.|x) \| \Ubar'(.|x)) \\
        &\overset{(b)}{\le} n \sup \limits_{P_X} D(U'\|\Ubar'|P_X).
    \end{align*}
    The inequality $(a)$ follows since $\mathbb{E}_{\Qshared}[\tau_n] \le n$ for a valid test. The inequality $(b)$ follows from the fact that $P_X(x) = \frac{\mathbb{E}[N_x]}{\mathbb{E}[\tau_n]}$ is a distribution on $\sX$ and we take a supremum over all possible distributions. We now consider the worst pair $U',\Ubar'$ that can be chosen by the adversary. Thus, we have
    \begin{align*}
        D(\Qshared|_{\mathcal{F}_{\tau_n}}\|\Qbarshared|_{\mathcal{F}_{\tau_n}}) &\le n \min_{\substack{U \in \convHull(\nullHyp) \\ \Ubar \in \convHull(\altHyp)}} \sup \limits_{P_X} D(U\|\Ubar|P_X) \\
        &= n \DOptSh.
    \end{align*}
    The final equality is because $\DOptSh$ is a saddle point and $\min$ and $\sup$ can be interchanged. The L.H.S. in \eqref{eqn:seq_shared_conv_dp} can be lower bounded as follows,
    \begin{align*}
        &D(\text{Bern}(\alpha(\phi_n,\hat{P}_S))\|\text{Bern}(1-\beta(\phi_n, \hat{P}_{\bar{S}})))  \\
        &= -h(\alpha(\phi_n,\hat{P}_S)) - \alpha(\phi_n,\hat{P}_S) \log (1-\beta(\phi_n, \hat{P}_{\bar{S}})) -(1-\alpha(\phi_n,\hat{P}_S)) \log \beta(\phi_n, \hat{P}_{\bar{S}}) \\
        &\ge -h(\alpha(\phi_n,\hat{P}_S)) -(1-\alpha(\phi_n,\hat{P}_S)) \log \beta(\phi_n, \hat{P}_{\bar{S}}).
    \end{align*}
    The last inequality holds because we drop a non-negative term. Thus, we have
    \begin{align*}
        -\frac{\log \beta(\phi_n, \hat{P}_{\bar{S}})}{n} \le \frac{\DOptSh + \frac{h(\alpha(\phi_n,\hat{P}_S))}{n}}{(1-\alpha(\phi_n,\hat{P}_S))}
    \end{align*}
    Since we assume that $E_0>0$, by definition of $E_0$ we have $\alpha(\phi_n,\hat{P}_S) \rightarrow 0$ as $n \rightarrow \infty$. Thus, we get
    \begin{equation} \label{eqn:seq_shared_conv_e1}
        \lim_{n \rightarrow \infty} -\frac{\log \beta(\phi_n, \hat{P}_{\bar{S}})}{n} \le \DOptSh.
    \end{equation}
    Now fix a different attack strategy. Under $H_0$ the adversary chooses $P_S$ i.i.d such that $\sum_{s}P_S(s)W(.|.,s)=U''$ and under $H_1$ it chooses $P_{\bar{S}}$ such that $\sum_{\bar{s}}P_{\bar{S}}\Wbar(.|.,\bar{s})=\Ubar''$. By approaching on similar lines, we get that if $E_1>0$, then
    \begin{equation} \label{eqn:seq_shared_conv_e0}
        \lim_{n \rightarrow \infty} \frac{\alpha(\phi_n,\hat{P}_S)}{n} \le \DOptShBar.
    \end{equation}
    Taken together, \eqref{eqn:seq_shared_conv_e1} and \eqref{eqn:seq_shared_conv_e0} complete the proof.
\end{proof}

\subsection{Deterministic}
A test $\phi$ is defined by the tuple $(\hat{P},\tau, Z)$, where the transmitter strategy $\hat{P}$ is a point mass on a sequence $(x_1,x_2,\cdots)$, $\tau$ is a stopping time of the filtration $\mathcal{F}_0 \subseteq \mathcal{F}_1 \cdots \subseteq \mathcal{F}_t \cdots \subseteq \mathcal{F}$ where $\mathcal{F}_t :=\ \sigma\{Y_1,\ldots,Y_t\}$, $Z:\mathcal{F}_{\tau}\to\{0,1\}$ is a $\mathcal{F}_{\tau}$-measurable function that specifies the decision rule applied by the detector. The definitions of errors and error exponents are analogous to the previous subsection.
Recall that
\begin{equation*}
    \DOptDet := \max \limits_{x} \min_{\substack{U_x \in \convHull(\nullHyp_x) \\ \Ubar_x \in \convHull(\altHyp_x)}} D(U_x\|\Ubar_x).   
\end{equation*}
Let
\begin{equation*}
    \DOptDetBar := \max \limits_{x} \min_{\substack{U_x \in \convHull(\nullHyp_x) \\ \Ubar_x \in \convHull(\altHyp_x)}} D(\Ubar_x\|U_x).
\end{equation*}
Recall that 
\begin{equation*}
\convHull(\nullHyp_x) \defineqq \inb{\sum \limits_{s \in \scriptS}P_{S}(s)W(.|x,s) : P_{S} \in \Delta_{\scriptS}},
\end{equation*}
$\convHull(\altHyp_x)$ is defined similarly with $\bar{S},\Wbar$ instead of $S,W$.
\begin{theorem}[] \label{thrm:det_weak_seq}
	Let $\nullHyp$ and $\altHyp$ be two sets of discrete memoryless channels which map $\mathcal{X}$ to $\mathcal{Y}$. The set of achievable pairs of exponents is given by
	\begin{equation} \label{eqn:det_weak_seq}
	\mathcal{E}_{\textup{det}}^{\textup{seq}}(\nullHyp,\altHyp) = \inb{ (E_0,E_1):E_0 \le \DOptDetBar, E_1 \le \DOptDet}.
	\end{equation}
\end{theorem}
\begin{proof}
    The proof of achievability is similar as in the case of shared randomness. We again invoke Lemma~\ref{lemma:seq_test} where $C_n^0$ and $C_n^1$ are fixed length schemes which achieve the Chenoff-Stein exponent in Theorem~\ref{thrm:det_weak}, $T_n$ is the same as either $C_n^0$ or $C_n^1$, $E_0=\DOptDetBar$ and $E_1=\DOptDet$. The proof of converse is also similar and works via the data processing inequality.
\end{proof}

\subsection{Private randomness}
A test $\phi$ is defined by the tuple $(\hat{P},\tau, Z)$, where $\hat{P}=P_{X_1}P_{X_2|X_1}P_{X_3|X_1,X_2}\cdots$ is the transmitter strategy, $\tau$ is a stopping time of the filtration $\mathcal{F}_0 \subseteq \mathcal{F}_1 \cdots \subseteq \mathcal{F}_t \cdots \subseteq \mathcal{F}$ where $\mathcal{F}_t := \sigma\{Y_1,\ldots,Y_t\}$, $Z:\mathcal{F}_{\tau}\to\{0,1\}$ is a $\mathcal{F}_{\tau}$-measurable decision function. The definitions of errors and error exponents are analogous to the shared randomness subsection. Let $\DOptPriv$ and $\DOptPrivBar$ be the (fixed length) Chernoff exponents for the private randomness case. Note that for this case, we do not have a single letter characterization for the exponents.

\begin{theorem}[] \label{thrm:priv_weak_seq}
	Let $\nullHyp$ and $\altHyp$ be two sets of discrete memoryless channels which map $\mathcal{X}$ to $\mathcal{Y}$. The corner point $(E_0=\DOptPrivBar,E_1=\DOptPriv) \in \mathcal{E}_{\textup{priv}}^{\textup{seq}}$.
\end{theorem}
\begin{proof}
    We again invoke Lemma~\ref{lemma:seq_test} for achievability. The scheme given in the achievability proof of Theorem~\ref{thm:pvt_rand} achieves a  positive exponent for both type-I and type-II errors. Thus, it can be used as $T_n$. $C_n^0$ and $C_n^1$ are fixed length schemes achieving the Chernoff-Stein exponent. Note that unlike in the previous cases, we do not have an explicit description of $C_n^0,C_n^1$ and use them as blackboxes. For this case, we do not have a converse argument.
\end{proof}

\subsection{Role of adaptivity in the sequential setting}
The proofs of achievability of a pair of exponents in the sequential setting work by invoking the achievability of the individual exponents in the fixed length setting (Lemma~\ref{lemma:seq_test}). Thus, the role of adaptivity is similar to the fixed length setting (refer to Section~\ref{sec:adaptive}). Thus, for the shared randomness and deterministic cases the achievability results hold even when the adversary is adaptive. For the private randomness case, the achievability results hold even when both the transmitter and adversary are adaptive. The converse proof in the shared randomness case works even when the transmitter is adaptive. In the deterministic case, the converse works when both the transmitter and adversary are adaptive.

\appendices

\section{Preliminaries} \label{app:prelim}

{\bf Adversarial Hypothesis Testing.}
Our achievability proofs use the adversarial Chernoff-Stein lemma and Chernoff information lemma from~\cite{brandao2020adversarial} which we briefly describe here. 
Let $\mathcal{Z}$ be a finite set. Let $\mathcal{P}, \mathcal{Q} \subseteq \mathbb{R}^{\mathcal{Z}}$ be closed, convex sets of probability distributions with a common support.
The adaptive adversary is specified by $\hat{p}_i : \mathcal{Z}^{i-1} \rightarrow \mathcal{P}$ and $\hat{q}_i : \mathcal{Z}^{i-1} \rightarrow \mathcal{Q}$ for $i\in[1:n]$.
For any $z^n \in \mathcal{Z}^n$, let $\hat{p}(z^n) := \prod_{i=1}^{n} \hat{p}_i(z^{i-1})(z_i)$ and $\hat{q}(z^n) := \prod_{i=1}^{n} \hat{q}_i(z^{i-1})(z_i)$. 
Let $A_n \subseteq \mathcal{Z}^n$ be an acceptance region for $\mathcal{P}$. For $\eps > 0$, the type-I and type-II errors are defined to be 
\begin{align*}
\alpha_n &\defineqq \sup_{\left(\hat{p}_i\right)_{i=1}^n} \hat{p}(A_n^c),
&\beta_n &\defineqq \sup_{\left(\hat{q}_i\right)_{i=1}^n} \hat{q}(A_n).
\end{align*}
The optimal type-II error when the type-I error is below $\eps$ is given by $\beta_n^{\eps} \defineqq \min_{A_n:\alpha_n \le \eps} \beta_n$. The adversarial Chernoff-Stein exponent is given by
\begin{equation*}
{\mathcal{E}}_{\textup{adv}}^{\epsilon}(\mathcal{P},\mathcal{Q}) \defineqq \lim_{n \rightarrow \infty} -\frac{1}{n} \log \beta_n^{\eps}.
\end{equation*}
For any pair $p \in \mathcal{P}, q \in \mathcal{Q}$, since the adversary may (non-adaptively) choose $\hat{p}_i=p$ and $\hat{q}_i=q$ for all $i\in[1:n]$, by the Chernoff-Stein lemma~\cite[Theorem~11.8.3]{cover1999elements} it is clear that $\advExp(\mathcal{P},\mathcal{Q}) \leq \min \limits_{p \in \mathcal{P}, q \in \mathcal{Q}} D(p\|q)$. In~\cite{fangwei1996hypothesis} it was shown that this upper bound is achievable if the adversary is non-adaptive. The following theorem states that this remains true even when the adversary is adaptive.
\begin{theorem}[Adversarial Chernoff-Stein Lemma~\cite{brandao2020adversarial}] \label{thrm:bhlp1}
	Let $\mathcal{Z}$ be a finite domain. For any pair of closed convex sets of probability distributions $\mathcal{P},\mathcal{Q} \subseteq \mathbb{R}^\mathcal{Z}$, 
	\begin{equation} \label{eqn:adv_chernoff_stein}
	\advExp(\mathcal{P},\mathcal{Q}) = \min \limits_{p \in \mathcal{P}, q \in \mathcal{Q}} D(p\|q).
	\end{equation}
 Let $(p_{\textup{CS}}^*,q_{\textup{CS}}^*) = \arg \min \limits_{p \in \mathcal{P}, q \in \mathcal{Q}} D(p\|q)$. The acceptance region which achieves the exponent in \eqref{eqn:adv_chernoff_stein} is given by
\begin{equation*}
    A_{n,\delta} = \inb{ z^n: \sum_{i=1}^{n} \log \frac{p_{\textup{CS}}^*(z_i)}{q_{\textup{CS}}^*(z_i)} \ge n(D(p_{\textup{CS}}^* \| q_{\textup{CS}}^*) - \delta) },
\end{equation*}
where $\delta > 0$. This ensures that
\begin{equation} \label{eqn:adv_chernoff_stein1}
    \hat{p}(A_{n,\delta}^c) \le \mathcal{O} \inp{\frac{1}{\delta^2 n}}
\end{equation}
for all $\hat{p}$. And 
\begin{equation} \label{eqn:adv_chernoff_stein2}
    \hat{q}(A_{n,\delta}) \le 2^{-n(D(p_{\textup{CS}}^* \| q_{\textup{CS}}^*) - \delta)}
\end{equation}
for all $\hat{q}$.
\end{theorem}

\section{Proof of Theorem~\ref{thrm:shared_strong}} \label{app:shared_strong}

Let $\mu_{XY}, \nu_{XY}$ be distributions on $\scriptX \times \scriptY$, $t \in \mathbb{R}$.
\begin{equation*}
\phiLarge(\mu_{Y} \| \nu_{Y}) \defineqq \sum_{\scriptY}\mu_{Y}^{1-t}\nu_{Y}^{t}
\end{equation*}
\begin{equation*}
\phiLarge(\mu_{Y|X} \| \nu_{Y|X} | \mu_{X}) \defineqq \mathbb{E}_{X \sim \mu_{X}}  \left[ \phiLarge(\mu_{Y|X} \| \nu_{Y|X}) \right]
\end{equation*}
$\phiSmall$ is defined to be $\log$ of the corresponding $\phiLarge$ quantity.

We construct a memoryless adversary strategy. Let $P_{S^n} = \prod_{i=1}^{n} P_{S_i}$, $P_{\bar{S}^n} = \prod_{i=1}^{n} P_{\bar{S_i}}$ where $P_{S_i}$ and $P_{\bar{S_i}}$ will be specified in course of the proof. Let $Q^{n}$ and $\bar{Q}^{n}$ denote the joint distributions on $\mathcal{X}^n \times \mathcal{Y}^n$ under $H_0$ and $H_1$ respectively. They are given by
	\begin{align} \label{eqn:shared_conv_q}
	Q^{n}(x^n,y^n) = 
	\prod_{i=1}^{n} \vec{Q}_{i}(x_i|x^{i-1},y^{i-1})
	\left( \sum_{s_i \in \scriptS} P_{S_i}(s_i)W(y_i|x_i,s_i) \right)
	\end{align}
	and
	\begin{align} \label{eqn:shared_conv_qbar}
	\bar{Q}^{n}(x^n,y^n) =
	\prod_{i=1}^{n} \vec{Q}_{i}(x_i|x^{i-1},y^{i-1})
	\left( \sum_{\bar{s}_i \in \scriptSbar} P_{\bar{S}_i}(\bar{s}_i)\Wbar(y_i|x_i,\bar{s}_i) \right).
	\end{align}
	Here, $\vec{Q}_{i}(x_i|x^{i-1},y^{i-1})$ denotes the transmitter strategy at the $i^{\textup{th}}$ timestep.
Define $\Qtilt^{i-1}$ to be
\begin{equation} \label{eqn:qtilt_shared}
\Qtilt^{i-1} = \frac{(Q^{i-1})^{1-t} (\bar{Q}^{i-1})^{t}}{\phiLarge(Q^{i-1} \| \bar{Q}^{i-1})}.
\end{equation}
From the definition of $\phiLarge(.\|.)$, we can see that $\Qtilt^{i-1}$ is a distribution on $\scriptX^{i-1} \times \scriptY^{i-1}$. Let $\tilde{Q}_{X_i}$ be the marginal on $X_i$ induced by $\Qtilt^{i-1} \cdot \vec{Q}_{i}$,
\begin{equation*}
\tilde{Q}_{X_i}(x_i) = \sum_{x^{i-1},y^{i-1}}\Qtilt^{i-1}(x^{i-1},y^{i-1}) \cdot \vec{Q}_{i}(x_i|x^{i-1},y^{i-1}).
\end{equation*}
Thus, we have
\begin{align*}
&\phiLarge(Q^{n} \| \bar{Q}^{n}) = \sum_{\domXY}(Q^{n})^{1-t}(\bar{Q}^{n})^{t} \\
&\overset{(a)}{=} \phiLarge(Q^{n-1} \| \bar{Q}^{n-1}) \sum_{\domXY} \Qtilt^{n-1} \vec{Q}_n (Q_{Y_n|X_n})^{1-t} (\bar{Q}_{Y_n|X_n})^{t} \\
&= \phiLarge(Q^{n-1} \| \bar{Q}^{n-1}) \cdot \phiLarge(Q_{Y_n|X_n} \| \bar{Q}_{Y_n|X_n} | \tilde{Q}_{X_n}),
\end{align*}
where $(a)$ follows from the factorizing $Q^{n}$ as $Q^{n}=Q^{n-1}\cdot\vec{Q}_n\cdot Q_{Y_n|X_n}$ and using \eqref{eqn:qtilt_shared}. We break down the term $\phiLarge(Q^{n-1} \| \bar{Q}^{n-1})$ in a similar manner. Repeating this process and finally taking $\log$ on both sides, we get
\begin{align*}
\phiSmall(Q^{n} \| \bar{Q}^{n}) &= \log \phiLarge(Q^{n} \| \bar{Q}^{n})\\
&= \sum_{i=1}^{n} \phiSmall(Q_{Y_i|X_i} \| \bar{Q}_{Y_i|X_i} | \tilde{Q}_{X_i})
\end{align*}
Define $\phiStarShared(t)$ to be
\begin{equation} \label{eqn:phiStarSharedDef}
\phiStarShared(t) \defineqq \sup \limits_{P_X} \min_{\substack{U \in \convHull(\nullHyp) \\ \Ubar \in \convHull(\altHyp)}} \phi_{t}(U \| \Ubar | P_{X}).
\end{equation}
We now specify $(P_{S_i},P_{\bar{S}_i})$ in the following manner. Consider the first term in the sum. By the definition of $\phiStarShared(t)$ in \eqref{eqn:phiStarSharedDef},
\begin{equation*}
\min_{P_{S_1},P_{\bar{S}_1}} \phiSmall(Q_{Y_1|X_1} \| \bar{Q}_{Y_1|X_1} | \tilde{Q}_{X_1}) \le \phiStarShared(t). 
\end{equation*}
Recall that $\phiSmall(.\|.)=-tD_{1-t}(.\|.)$ for $t<0$, where $D_{1-t}(.\|.)$ is the R\'enyi divergence of order $1-t$. Since $\mathcal{P}=\{ P_X U: U \in \convHull(\nullHyp) \}$, $\mathcal{Q} = \{ P_X \Ubar: \Ubar \in \convHull(\altHyp) \}$ are closed, convex sets and $D_{1-t}(.\|.)$ is lower semi-continuous \cite[Theorem 15]{van2014renyi}, such a minimum exists. We choose $(P_{S_1},P_{\bar{S}_1})$ such that $\phiSmall(Q_{Y_1|X_1} \| \bar{Q}_{Y_1|X_1} | \tilde{Q}_{X_1}) \le \phiStarShared(t)$. We now recursively specify all the $(P_{S_i},P_{\bar{S}_i})$ in a similar manner. Thus, we have
\begin{equation} \label{eqn:shared_conv_phi_ub}
\phiSmall(Q^{n} \| \bar{Q}^{n}) = \log \phiLarge(Q^{n} \| \bar{Q}^{n}) \le n\phiStarShared(t).
\end{equation} 
We now follow the approach of \cite[Section VI]{hayashi2009discrimination}, \cite{nagaoka2005strong}. Let $\typeOneErrRest$ and $\typeTwoErrRest$ be the type-1 and type-2 errors once the strategies of transmitter, detector and adversary are fixed. They are as defined in the Appendix~\ref{app:shared_strong}. Let 
\begin{equation*}
r \defineqq \liminf_{n \rightarrow \infty} \frac{-1}{n}\log \typeTwoErrRest
\end{equation*}
Our goal is to show that if $r > \DOptSh$, then then the type-1 error probability $\typeOneErrRest$ goes to $1$ exponentially fast.
As before the distribution of the decision is $\bern(\typeOneErrRest)$ under $H_0$ and $\bern(1-\typeTwoErrRest)$ under $H_1$. Since data processing inequality holds for $D_{1-t}(.\|.)$ for $t < 0$ \cite[Theorem 9]{van2014renyi}, we can apply it for $\phiLarge(.\|.)$.
\begin{align*}
\phiLarge(\bern(\typeOneErrRest) \| \bern(1-\typeTwoErrRest)) &\le \phiLarge(Q^{n} \| \bar{Q}^{n}) = e^{\phiSmall(Q^{n} \| \bar{Q}^{n})}
\end{align*}
Expanding out the L.H.S. and using \eqref{eqn:shared_conv_phi_ub}, we have
\begin{align*}
{(1-\typeOneErrRest)}^{1-t}{(\typeTwoErrRest)}^{t}  + {(\typeOneErrRest)}^{1-t}{(1-\typeTwoErrRest)}^{t} \le e^{n\phiStarShared(t)}.
\end{align*}
Since ${\typeOneErrRest}^{1-t}{(1-\typeTwoErrRest)}^{t} \ge 0$, it can be dropped while retaining the inequality. Taking $\log$ followed by $\liminf$ on both sides, we get
\begin{align*}
\liminf_{n \rightarrow \infty} -\frac{1}{n} \log (1-\typeOneErrRest) &\ge \frac{-tr-\phiStarShared(t)}{1-t} \\
&\ge \sup \limits_{t<0} \frac{-t}{1-t} \left( r-\frac{\phiStarShared(t)}{-t} \right).
\end{align*}
We now show that the L.H.S. $>0$ for some choice of $t<0$.
\begin{align*}
\lim_{t \rightarrow 0^{-}} \frac{\phiStarShared(t)}{-t} 
&\overset{(a)}{=} \sup \limits_{P_X} \min_{\substack{U \in \convHull(\nullHyp) \\ \Ubar \in \convHull(\altHyp)}} \lim_{t \rightarrow 0^{-}} \frac{\phi_{t}(U \| \Ubar | P_{X})}{-t} \\
&\overset{(b)}{=} \sup \limits_{P_X} \min_{\substack{U \in \convHull(\nullHyp) \\ \Ubar \in \convHull(\altHyp)}} D(U \| \Ubar | P_{X}) \overset{(c)}{=} \DOptSh.
\end{align*}
where $(a)$ is by the definition of $\phiStarShared$ in \eqref{eqn:phiStarSharedDef} and the assumption in \eqref{eqn:shared_conv_assumption}, $(b)$ follows from the fact that $\frac{\phi_{t}(U \| \Ubar | P_{X})}{-t} = D_{1-t}(U \| \Ubar | P_{X})$ when $t<0$ and by the continuity $D_{1-t}$ in $t$ \cite{van2014renyi}, $(c)$ by the definition of $\DOptSh$ \eqref{eqn:DoptShDef}. Since $r > \DOptSh$, we have $r-\frac{\phiStarShared(t')}{-t'} > 0$ for some $t' < 0$.
\begin{align*}
\liminf_{n \rightarrow \infty} -\frac{1}{n} \log (1-\typeOneErrRest) > 0
\end{align*}
This inequality holds true for all possible transmitter and detector strategies $(\vec{Q}, A_n)$. Thus, the probability of correctness under $H_0$ decays exponentially.


\section{Proof of Theorem~\ref{thrm:det_weak} (No Feedback)} \label{app:det_weak_nofb}
\paragraph{Achievability ($\advExpDet(\nullHyp,\altHyp) \ge \DOptDet$)}
We apply the same argument given in the achievability proof of Theorem~\ref{thrm:shared_weak} for a fixed choice of $x$. We then optimize over $x$ to complete the proof.

\paragraph{Converse ($\advExpDet(\nullHyp,\altHyp) \le \frac{\DOptDet}{1-\eps}$)}
Recall that transmitter strategy is a fixed tuple $(x_1,x_2,\ldots,x_n)$.
Consider a memoryless adversary strategy. Let $Q^n$ (resp. $\bar{Q}^n$) be the distribution induced on $\scriptY$ under $H_0$ (resp. $H_1$). In this setting, $D(Q^n \| \bar{Q}^n) = \sum_{i=1}^{n}D(Q_{Y_i}\|\bar{Q}_{Y_i})$, where $Q_{Y_i}, \bar{Q}_{Y_i}$ are the marginals on $Y_i$ under $H_0$ and $H_1$ respectively. It is easy to see that each term in the sum is upper bounded by $\DOptDet$. Thus, $D(Q^n \| \bar{Q}^n) \le n\DOptDet$. The rest of the proof then follows from the data processing inequality (e.g., see \cite[Section VI]{hayashi2009discrimination}).

\paragraph{Strong Converse} The proof is on similar to the proof of Theorem~\ref{thrm:shared_strong} (Appendix~\ref{app:shared_strong}).

\paragraph{Characterization ($\advExpDet(\nullHyp,\altHyp)>0 \iff \convHull(\nullHyp_x) \cap \convHull(\altHyp_x)=\emptyset$) for some $x$} The  if ($\Leftarrow$) part follows from Theorem~\ref{thrm:det_weak}. Consider the contrapositive of the only if ($\Rightarrow$) direction. Observe that under hypothesis $H_0$ (resp., $H_1$), the adversary may induce any conditional distribution $\convHull(\nullHyp_x)$ (resp., $\convHull(\altHyp_x)$) when the transmitter sends the symbol $x$. Hence, when the intersection is non-empty for all $x$, the adversary may induce the same conditional distribution under both hypotheses so that no transmission strategy (including an adaptive one) can distinguish between the hypotheses.

\section{Proof of Theorem~\ref{thrm:det_weak} (Feedback to Transmitter and Adversary)} \label{app:det_weak_fb}
The proof of achievability is same as Appendix~\ref{app:det_weak_nofb}.

\paragraph*{Converse ($\advExpDet(\nullHyp,\altHyp) \le \frac{\DOptDet}{1-\eps}$)}
We restrict the adversary to choose the next state independently conditioned on the previous outputs of the channel, i.e. $P_{S_i|S^{i-1},Y^{i-1}} = P_{S_i|Y^{i-1}}$, $P_{\bar{S}_i|\bar{S}^{i-1},Y^{i-1}} = P_{\bar{S_i}|Y^{i-1}}$ where $P_{S_i|Y^{i-1}}$ and $P_{\bar{S_i}|Y^{i-1}}$ will be specified in course of the proof. The transmitter strategy is given by a set of deterministic functions $\{ g_i: \scriptY^{i-1} \rightarrow \scriptX \}$, where $g_1$ is a constant function with value $x_1$. Let $Q^{n}$ and $\bar{Q}^{n}$ denote the joint distributions on $\mathcal{Y}^n$ under $H_0$ and $H_1$ respectively. $Q^{n}$ is given by
\begin{align} \label{eqn:det_fb_conv_q}
Q^{n}(y^n) = \prod_{i=1}^{n} \big( \sum_{s_i \in \scriptS} P_{S_i|Y^{i-1}}(s_i|y^{i-1}) W(y_i|g_i(y^{i-1}),s_i) \big).
\end{align}
$\bar{Q}^n$ is defined similarly with $\bar{S},\Wbar$. We again try to upper bound $D(Q^n \| \bar{Q}^n)$. 
\begin{equation} \label{eqn:det_fb_conv_kl}
D(Q^n \| \bar{Q}^n) = \sum_{i=1}^{n}D(Q_{Y_i|Y^{i-1}} \| \bar{Q}_{Y_i|Y^{i-1}}|Q_{Y^{i-1}})
\end{equation}
Consider the $i^{\textup{th}}$ term in \eqref{eqn:det_fb_conv_kl}. For each tuple $(y^{i-1})$, by the definition of $\DOptDet$ in \eqref{eqn:DoptDetDef}, we have
\begin{align} \label{eqn:det_min_condition}
\min_{\substack{P_{S_i|Y^{i-1}}(.|y^{i-1}) \\ P_{\bar{S}_i|Y^{i-1}}(.|y^{i-1})}} D(Q_{Y_i|Y^{i-1}}(.|y^{i-1}) \| \bar{Q}_{Y_i|Y^{i-1}}(.|y^{i-1}))
\le \DOptDet.
\end{align}
For each tuple $(y^{i-1})$, we specify $P_{S_i|Y^{i-1}}(.|y^{i-1})$ and $P_{\bar{S}_i|Y^{i-1}}(.|y^{i-1})$ such that they satisfy \eqref{eqn:det_min_condition}.\\ Since $D(Q_{Y_i|Y^{i-1}} \| \bar{Q}_{Y_i|Y^{i-1}}|Q_{Y^{i-1}})$ is an averaging over $y^{i-1}$, it is also upper bounded by $\DOptDet$. Repeating this argument for each term in the sum \eqref{eqn:det_fb_conv_kl}, we get $D(Q^n \| \bar{Q}^n) \le n\DOptDet$. The rest of the proof is similar to Theorem~\ref{thrm:shared_weak}.

\section{Role of Memory for a Privately Randomized Transmitter: Adapative Adversary Case} \label{app:example1}

Continuing the discussion from Example~\ref{ex:role_of_memory}, we now allow the adversary access to feedback, i.e. its choice of state can depend on the outputs of the previous transmission. The new state spaces for the adversary are $\scriptS^2 = \{0\}$ and $\scriptSbar^2 = \scriptSbar \times \Sigma$ where $\Sigma = \{ \sigma: \scriptY \rightarrow \{ 0,1 \} \}$. Observe that $\Sigma$ accounts for feedback. Note that $|\scriptSbar^2|=2 \times |\Sigma|=16$. The problem can now be thought of as a new hypothesis test between \\
$H_0$ : $\nullHyp^2=\{ W^2(.|.) \}$ where
\begin{equation*}
W^2((y_1,y_2)|(x_1,x_2)) = W(y_1|x_1)W(y_2|x_2)
\end{equation*}
and $H_1$ : $\altHyp^2 = \{ \Wbar^2(.|.,(\bar{s},\sigma)) :(\bar{s},\sigma) \in \scriptSbar \times \Sigma \}$ where
\begin{equation*}
\Wbar^2((y_1,y_2)|(x_1,x_2),(\bar{s},\sigma)) \\
= \Wbar(y_1|x_1,\bar{s}) \Wbar(y_2|x_2,\sigma(y_1)).
\end{equation*}
Recall that the transmitter strategy was $P_{X_1,X_2}(0,0)=P_{X_1,X_2}(1,1)$.
The adversary strategy is given by $P_{\bar{S},\sigma}$. Let $\mathcal{Q}$ (resp. $\bar{\mathcal{Q}}$) be the set of all possible (double-letter) distributions that can be induced on $\scriptY^{2}$ under $H_0$ (resp. $H_1$). Since $\mathcal{Q}$ is a singleton, let the member be denoted by $Q_{Y_1,Y_2}$.
If $\mathcal{Q} \cap \bar{\mathcal{Q}} = \emptyset$, then by Theorem \ref{thrm:bhlp1}, we get a positive exponent. Assume for contradiction that this is not the case, i.e. there exists $P_{\bar{S},\sigma}$  such that the resulting $\bar{Q}_{Y_1,Y_2}$ is same as $Q_{Y_1,Y_2}$. Since the marginals have to be equal, we have $Q_{Y_1}=\bar{Q}_{Y_1}$. This forces $P_{\bar{S}}$ to be uniform. Now, observe that $Q_{Y_1,Y_2}(0,1)=0$. Examine the term corresponding to $x_1=x_2=1,\bar{s}_1 = 1$  in the expansion of $\bar{Q}_{Y_1,Y_2}(0,1)$.
\begin{align*}
P_{X_1,X_2}(1,1)P_{\bar{S}}(1) \sum_{\sigma_2 \in \Sigma}P_{\sigma|\bar{S}}(\sigma_2|1) \Wbar(0|1,1) \Wbar(1|1,\sigma_{2}(0)) \\
\end{align*}
It cannot be zero since $\Wbar(0|1,1) > 0$, $\Wbar(1|1,\sigma_{2}(0)) > 0$ for all $\sigma_2$ when $0<r<1$. Thus, we have a contradiction. This scheme gets us a positive exponent even when the adversary is adaptive.

\newcommand{\bS}{\bar{S}}
\newcommand{\bs}{\bar{\vecs}}

\section{Proof of Theorem~\ref{thm:pvt_rand}}\label{app:pvt_rand}
We use bold faced letters to denote $n$-length vectors, for example, $\vecx$ denotes a vector in $\cX^n$ and $\vecX$ denotes a random vector taking values in $\cX^n$. 
For a random variable $X$, we denote its distribution by $P_X$ and use the notation $X\sim P_{X}$ to indicate this. 
For an alphabet $\cX$, let $\mathcal{P}^n_{\cX}$ denote the set of all empirical distributions of $n$ length strings from $\cX^n$. 
For a random variable $X\sim P_{X}$ such that  $P_X \in \mathcal{P}^n_{\cX}$, let $\cT^n_X$ be the set of all $n$-length strings with empirical distribution $P_X$. 
For $\vecx\in \cX^n$, the statement $\vecx\in \cT^n_{X}$ defines $P_{X}$ as the empirical distribution of $\vecx$ and a random variable $X\sim P_{X}$. For $P_{XY}\in \mathcal{P}^n_{\cX\times\cY}$ and $\vecy\in \cY^n$, let $\cT^n_{X|Y}(\vecy)= \inb{\vecx: (\vecx, \vecy)\in \cT^n_{XY}}$.
We denote $2^a$ by $\exp\inp{a}$.

\begin{definition}\label{defn:eta_p}
For a distribution $P$ over $\cX$, we define $\eta(P)$ as the set of triples $(\eta_1, \eta_2, \eta_3)$ for which there exists $\delta>0$ 
such that there is no joint distribution $P_{XX'\bar{S}SY}$ with $P_X=P_{X'} = P$ satisfying 
\begin{enumerate}
	\item $I(X;\bar{S})< \eta_1$,
	\item $I(X';S)< \delta$,
	\item $D(P_{X\bar{S}Y}||P_{X\bar{S}}W)<\eta_2$,
	\item $D(P_{X'SY}||P_{X'S}W)<\delta$, and
	\item if $P_{XX'}(X'\neq X)>0$, \begin{enumerate} \item[(i)] $I(X';XY|\bar{S})<\eta_3$, and \item[(ii)] $I(X;X'Y|S)<\delta$. \end{enumerate}
\end{enumerate}  
\end{definition}
We will first prove the following lemma.
\begin{lemma}\label{app_lemma:disambiguity}
If $\inp{\nullHyp,\altHyp}$ is not trans-symmetrizable and $\convHull(\nullHyp)\cap\convHull(\altHyp)=\emptyset$, then there exists an input distribution $P$ with $(\eta_1, \eta_2, \eta_3)\in \eta(P)$ such that $\eta_1, \eta_2, \eta_3>0$. In particular, for uniform distribution $P$ on the input alphabet $\cX$,  $\eta_1 = \eta_2 = \min\inb{\frac{\zeta_1^2}{5|\cX|^2}, \frac{\zeta_2^2}{11|\cX|^4}}$ and $\eta_3 = \frac{3\zeta_2^2}{11|\cX|^4}$.
\end{lemma}
\begin{proof}[Proof of Lemma~\ref{app_lemma:disambiguity}]
We show that if a pair of channels $\inp{\nullHyp,\altHyp}$ is  not trans-symmetrizable and  $\convHull(\nullHyp)\cap\convHull(\altHyp)=\emptyset$, then for any full support input distribution $P$, there exist $(\eta_1, \eta_2, \eta_3)\in \eta(P)$ such that $\eta_1, \eta_2, \eta_3>0$. 

Recall from \eqref{eq:gap_empty_intersection} and \eqref{eq:gap_to_sym} that if $\convHull(\nullHyp)\cap\convHull(\altHyp)=\emptyset$, there exists a constant $\zeta_1>0$ such that for every $P_{\bar{S}}$ on $\bar{\cW}$ and $P_{S}$ on $\cW$, 
 \begin{align*}
 &\max_{x, y}\inl{\sum_{\bar{s}}P_{\bar{S}}(\bar{s})\bar{W}(y|x, \bar{s})- \sum_{s}P_{S}(s)W(y|x, s)}>\zeta_1
\end{align*}
and if $\inp{\nullHyp,\altHyp}$ is not trans-symmetrizable, there exists $\zeta_2>0$ such that for every $P_{S|X}(s|x'),\, s\in \cS, x'\in \cX$ and $P_{\bar{S}|X}(\bar{s}|x), \bar{s}\in \bar{\cS}, x\in \cX$,
\begin{align*}
&\max_{x, x', y}\inl{\sum_{s \in \scriptS}P_{S|X}(s|x')W(y|x,s)-\sum_{\bar{s} \in \scriptSbar}P_{\bar{S}|X}(\bar{s}|x)W(y|x',\bar{s})}\nonumber\\
&\qquad>\zeta_2.
\end{align*}

We consider a full support input distribution $P$. That is, $\min_xP(x)\geq \alpha$ for some $\alpha>0$. We will show that there exists $(\eta_1, \eta_2, \eta_3)\in \eta(P)$ such that $\eta_1, \eta_2, \eta_3>0$ for some $\delta>0$. These choices only depend on $\alpha$, $\zeta_1$ and $\zeta_2$.

Suppose, for the sake of contradiction, for every $\eta_1,\, \eta_2>0$, $\eta_3>0$, there exists a $P_{XX'\bar{S}SY}$ such that for $(X,X')\sim P_{XX'}$,  $P_{XX'}\inp{X \neq X'}>0$ and conditions 1), 2), 3), 4) and 5) hold in Definition~\ref{defn:eta_p}.  We have, for $\bar{W}  = W_{Y|X\bar{S}}$,
 \begin{align*}
 \eta_1+\eta_2&+\eta_3>I(X;\bar{S}) + D(P_{X\bar{S}Y}||P_{X\bar{S}}\bar{W})+ I(X';XY|\bar{S})\\
 &= D(P_{XX'\bar{S}Y}||P_{X}P_{X'\bar{S}}W_{Y|X\bar{S}})\\
 &\geq D(P_{XX'Y}||\sum_{\bar{s}}P_{X}P_{X'}P_{\bar{S}|X'}(\bar{s}|\cdot)W_{Y|X\bar{S}}(\cdot|\cdot, \bar{s})).
 \end{align*}Using Pinsker's inequality, this implies that 
 \begin{align}\label{app_eq4}
 &\sum_{x, x', y}\Big|P_{XX'Y}(x, x', y)\nonumber\\
 &-\sum_{\bar{s}}P_{X}(x)P_{X'}(x')P_{\bar{S}|X'}(\bar{s}|x'){W_{Y|X\bar{S}}}(y|x, \bar{s})\Big|\nonumber\\&\leq \sqrt{2\ln{2}}\sqrt{\eta_1+\eta_2+\eta_3}.
 \end{align}
 Similarly, from conditions 2), 4) and 5) (ii) in Definition~\ref{defn:eta_p}, we can write
  \begin{align}\label{app_eq5}
 &\sum_{x, x', y}\Big|P_{XX'Y}(x, x', y) -\nonumber\\
 &-\sum_{s}P_{X}(x)P_{X'}(x')P_{S|X}(s|x)W_{Y|X' S}(y|x', s)\Big|\leq \sqrt{2\ln{2}}\sqrt{3\delta}.
 \end{align}Combining \eqref{app_eq4} and \eqref{app_eq5} and noting that $\ln{2}\leq 1$, we obtain
 \begin{align}
 &\sum_{x, x', y}P_{X}(x)P_{X'}(x')\Big|\sum_{\bar{s}}P_{\bar{S}|X'}(\bar{s}|x')W_{Y|X\bar{S}}(y|x, \bar{s})\nonumber-\\ &\sum_{s}P_{S|X}(s|x)W_{Y|X' S}(y|x', s)\Big| \leq \sqrt{2}\inp{\sqrt{\eta_1+\eta_2+\eta_3}+\sqrt{3\delta}}.
 \end{align}
 This implies that
 \begin{align}
 &\max_{x, x', y}\Big|\sum_{\bar{s}}P_{\bar{S}|X'}(\bar{s}|x')W_{Y|X\bar{S}}(y|x, \bar{s})\nonumber\\ &- \sum_{s}P_{S|X}(s|x)W_{Y|X' S}(y|x', s)\Big| \leq\frac{ \sqrt{2}\inp{\sqrt{\eta_1+\eta_2+\eta_3}+\sqrt{3\delta}}}{\alpha^2}	
 \end{align} which is a contradiction to \eqref{eq:gap_to_sym} for $\eta_1, \eta_2,$ and $\eta_3$ satisfying 
 \begin{align}\label{app_eq6}
 \frac{ \sqrt{2}\sqrt{\eta_1+\eta_2+\eta_3}+\sqrt{6\delta}}{\alpha^2}\leq \zeta_2
 \end{align}for some $\delta>0$.
Next, suppose that there exists $P_{XX'\bar{S}SY}$ such that for $(X,X')\sim P_{XX'}$,  $P_{XX'}\inp{X = X'} =1$ such that conditions 1), 2), 3) and 4) hold in Definition~\ref{defn:eta_p}. 
Setting $X' = X$ and proceeding in a similar manner, one can show that
\begin{align*}
 &\max_{x, y}\inl{\sum_{\bar{s}}P_{\bar{S}}(\bar{s})W_{Y|X\bar{S}}(y|x, \bar{s})- \sum_{s}P_{S}(s)W_{Y|XS}(y|x, s)}\nonumber\\&\leq \frac{\sqrt{2}\sqrt{\eta_1+\eta_2}+\sqrt{4\delta}}{\alpha}
 \end{align*} which is a contradiction to \eqref{eq:gap_empty_intersection} for 
 \begin{align}
 \frac{\sqrt{2}\sqrt{\eta_1+\eta_2}+\sqrt{4\delta}}{\alpha}\leq \zeta_1.\label{app_eq3}
 \end{align}

 Since, $\zeta_1$ and $\zeta_2$ are both positive, we can choose $\eta_1, \eta_2, \eta_3>0$ such that for some $\delta>0$, \eqref{app_eq3} and \eqref{app_eq6} hold. Note that such a choice only depends on $\alpha$, $\zeta_1$ and $\zeta_2$.

In particular, if $P$ is uniform, then $\alpha = \frac{1}{|\cX|}$. If we choose $\eta_1 = \eta_2 = \min\inb{\frac{\zeta_1^2}{5|\cX|^2}, \frac{\zeta_2^2}{11|\cX|^4}}$ and $\eta_3 = \frac{3\zeta_2^2}{11|\cX|^4}$. Then, \eqref{app_eq3} and \eqref{app_eq6} hold for some $\delta>0$.
\end{proof}

\begin{proof}[Proof of Theorem~\ref{thm:pvt_rand}]
We already discussed (after Definition~\ref{defn:trans-sym}) how trans-symmetrizability implies $\bpvtexp = 0$. Here, we provide the proof of the lower bound of \begin{align*}\min\inb{\eta_1, \eta_2,\eta_3/3}
\end{align*} on the exponent. This combined with Lemma~\ref{app_lemma:disambiguity} will give us the lower bound on the exponent in the theorem statement. 

The proof uses the method of types (See \cite{csiszar1998method,csiszar2011information}).
We recall some properties from \cite[Chapter~2]{csiszar2011information}. Let $X$ and $Y$ be two jointly distributed random variables according to a joint type $P_{XY}\in \cP_n(\cX\times \cY)$. For $(x^n, y^n)\in \cT^n_{XY}$, a distribution $Q$ on $\cX$ and a discrete memoryless channel $U$ from $\cX$ to $\cY$, we have
\begin{align}
|\cP_n(\cX)|&\leq (n+1)^{|\cX|}\label{eq:type1}\\
(n+1)^{-|\cX|}\exp\inp{nH(X)}&\leq \cT^n_{X}\leq \exp\inp{nH(X)}\label{eq:type2}\\
(n+1)^{-|\cX||\cY|}\exp\inp{nH(Y|X)}&\leq \cT^n_{Y|X}(\vecx)\leq \exp\inp{nH(Y|X)}\label{eq:type3}\\
(n+1)^{-|\cX|}\exp\inp{-nD(P_X||Q)}&\leq \sum_{\tilde{\vecx}\in \cT^n_{X}}Q^n(\tilde{\vecx}) \leq \exp\inp{-nD(P_X||Q)}\label{eq:type4}\\
\sum_{y\in \cT^n_{Y|X}(\vecx)}U^n(\vecy|\vecx) \leq \exp&\inp{-nD(P_{XY}||P_XU)}\label{eq:type5}.
\end{align}

For a distribution $P$, $(\eta_1, \eta_2, \eta_3)\in \eta(P)$ and $R=\eta_3/3$, we first show that we can obtain an exponent $\gamma$ (see \eqref{eq:exponent_intermediate} below) for the probability of error under Hypothesis $H_1$. For any $\epsilon>0$,
\begin{align}
\gamma &\geq \min\Bigg\{\min_{\substack{P_{X\bS}:\\I(X;\bS)\geq \eta_1}}A_1,\eta_2-\epsilon, \min_{\substack{P_{X\bar{S}X'SY}:\\I(X';XY|\bar{S})\geq \eta_3}}A_2\Bigg\}\label{eq:exponent_intermediate}
\end{align}
\begin{align}
\text{ where }&A_1 =
R-\inps{R-I(X;\bar{S})}-\epsilon\qquad\text{ and }\label{eq:exponent_A1}\\
A_2 &= \max\Big\{I(X;X'\bar{S})-\inps{R-I(X';\bar{S})}-\epsilon,I(Y;X'|X\bar{S})-\inps{R-{I(X';X\bar{S})}}-2\epsilon\Big\}\label{eq:exponent_A3}
\end{align}
For $N = \exp\inp{nR}$, let $\cC(P) = \inb{\vecx_1,  \ldots, \vecx_N}$ be a set of sequences of type $P$ given by Lemma~\ref{app_lem:codebook} (proved on page~\pageref{eq:proof_lem_codebook}). The lemma is based on \cite[Lemma 3]{csiszar1988capacity}.
\begin{lemma}\label{app_lem:codebook}
For any $\epsilon>0$, large enough $n$, $N := 2^{nR}$ for $R\geq \epsilon$, and type $P$, there exist sequences $\vecx_1, \ldots, \vecx_N\in \cX^n$ of type $P$, such that for every $\vecx\in \cX^n$, $\vecs\in \cS^n\cup\bar{\cS}^n$ and every joint type $P_{XX'S}$, we have 
\begin{align}
&\inl{\inb{j:(\vecx, \vecx_j, \vecs)\in {\cT^n_{XX'S}}}}\leq \exp\inb{n\inp{\inps{R-I(X';XS)}+\epsilon}},\label{eq:codebook1}
\\
&\frac{1}{N}\inl{\inb{i:\inp{\vecx_i, \vecs}\in \cT^n_{XS}}}\leq \exp\inb{n\inp{\inps{R-I(X;S)}-R+\epsilon/2}}\label{eq:codebook2},\text{ and}\\
&\frac{1}{N}\inl{\inb{i:(\vecx_i, \vecx_j, \vecs)\in \cT^n_{XX'S}\text{ for some }j\neq i}}\leq \exp\inb{n\inp{\inps{R-I(X';S)}-I(X;X'S)+\epsilon/2}}\label{eq:codebook3}
\end{align}
\end{lemma}
The transmitter sends an input sequence selected uniformly at random (using its private randomness) from $\vecx_1, \vecx_2, \ldots, \vecx_N$. 
\begin{definition}[Detector]\label{defn:detector}
Given sequences $\inb{\vecx_1, \ldots, \vecx}$, each of type $P$, and for $(\eta_1, \eta_2,\eta_3)\in \eta(P)$ and $\delta>0$ given by Definition~\ref{defn:eta_p}, $\phi(\vecy) = H_1$ if and only if there exist $i\in [1:N]$ and  ${\bar{\vecs}}\in \bar{\cS}^n$  such that for the joint empirical distribution $P_{X\bar{S}Y}$ of  $(\vecx_i, {\bar{\vecs}}, \vecy)$, 
\begin{enumerate}
    \item $I(X;\bar{S})< \eta_1$
    \item $D(P_{X\bar{S}Y}||P_{X\bar{S}}\bar{W})<\eta_2$, and 
    \item for each $j$ such that the joint empirical distribution $P_{X\bar{S}X'SY}$ of $(\vecx_i, {\bar{\vecs}},\vecx_j, \vecs, \vecy)$  for some $\vecs\in {\cS}^n$ satisfies $I(X';S)< \delta$ and $D(P_{X'SY}||P_{X'S}W)<\delta$, we have $I(X';XY|\bar{S})<\eta_3$.
\end{enumerate}
\end{definition}

Suppose the active hypothesis is $H_1$. Firstly, notice that the probability of error under any randomized attack can be written as an average over deterministic attacks and is thus maximized by a deterministic attack. So, it is sufficient to consider only deterministic attacks by the adversary. Suppose the adversary attack sequence is $\bar{\vecs}\in \bar{\cS}^n$.

Let $P_{\vecX\vecY}(\vecx_i, \vecy) = \frac{1}{N}W^n(\vecy|\vecx_i, \bs)$ for $\vecx_i\in \cC(P)$, $\vecy\in \cY^n$ and $P_{\vecX\vecY}(\vecx, \vecy) = 0$ for $\vecx\notin\cC(P)$. Let $(\vecX, \vecY)\sim P_{\vecX\vecY}$. 
Define events
\begin{align*}
\cE_1 &:= \inb{(\vecX, \bs)\in \cT^n_{X\bS} \text{ such that } I(X;\bS)\geq \eta_1},\\
\cE_2 &:= \inb{(\vecX, \bs, \vecY)\in \cT^n_{X\bS Y} \text{ such that }  D(P_{X\bS Y}||P_{X\bS}\times \bar{W})\geq \eta_2},
\end{align*}  
$\cE_3 := \big\{(\vecX, \bs, \vecY)\in \cT^n_{X\bS Y} \text{ such that }  I(X;\bS)< \eta_1,\\ D(P_{X\bS Y}||P_{X\bS}\times \bar{W})< \eta_2, \exists \vecx_j\neq \vecX \text{ such that }(\vecx_j, \vecs, \vecY)\in \cT^n_{X'SY} \text{ for some }\vecs\in {\cS}^n  \text{ for which }I(X';S)< \delta \\\text{ and }D(P_{X'SY}||P_{X'S}W)<\delta, \text{ but } I(X';XY|\bar{S})\geq \eta_3\big\}$, and
$\cE_4 := \big\{\exists\vecs\in {\cS}^n \text{ such that }(\vecX, \bs, \vecs, \vecY)\in \cT^n_{X\bS S Y} \\ \text{ for which }  I(X;\bS)< \eta_1,D(P_{X\bS Y}||P_{X\bS}\times \bar{W})< \eta_2, I(X;S)< \delta \text{ and }D(P_{XSY}||P_{XS}W)<\delta\big\}$. \\
Then, 
\begin{align*}
P_{\vecX \vecY}&\inp{\phi(\vecY) = H_0}\leq P_{\vecX \vecY}\inp{\cE_1\cup\cE_2\cup\cE_3\cup\cE_4}\\
&\leq P_{\vecX \vecY}\inp{\cE_1} + P_{\vecX \vecY}\inp{\cE_2} + P_{\vecX \vecY}\inp{\cE_3}+ P_{\vecX \vecY}\inp{\cE_4}
\end{align*}
We first note that $P_{\vecX \vecY}\inp{\cE_4} = 0$ because for $(\eta_1, \eta_2, \eta_3)\in \eta(P)$ and $\delta>0$ given by Definition~\ref{defn:eta_p}, there is no distribution $\cT^n_{XX'\bS S Y}$ (with $X' = X$) satisfying the conditions in $\cE_4$. Next, we  
evaluate $P_{\vecX \vecY}\inp{\cE_1}$, 
\begin{align}
&\bbP_{\vecX \vecY}\inp{(\vecX, \bar{\vecs})\in \cT^n_{X\bar{S}}, I(X;\bar{S})\geq \eta_1}\nonumber\\
&  =  \frac{\inl{i:(\vecx_i, \bar{\vecs})\in \cT^n_{X\bar{S}}, I(X;\bar{S})\geq \eta_1}}{N}\nonumber\\
& = \sum_{P_{X\bar{S}}\in \cP^n_{\cX\times\bar{\cS}}:I(X;\bS)\geq \eta_1}\frac{\inl{i:(\vecx_i, \bar{\vecs})\in \cT^n_{X\bar{S}}}}{N}\nonumber\\
&\stackrel{(a)}{\leq} \sum_{P_{X\bar{S}}:I(X;\bS)\geq \eta_1}\exp\inb{n\inp{\inps{R-I(X;\bS)}-R+\epsilon/2}}\nonumber\\
&\stackrel{(b)}{\leq} \max_{P_{X\bar{S}}:I(X;\bS)\geq \eta_1}\exp\inb{-n\inp{R-\inps{R-I(X;\bS)}-\epsilon}}\label{app_eq:exponent1}
\end{align} Here, $(a)$ holds because of \eqref{eq:codebook2} and $(b)$ holds for large $n$ as the number of joint types is at most polynomial in $n$ (see \eqref{eq:type1}). 
Next, we evaluate $P_{\vecX \vecY}\inp{\cE_2}$, 
\begin{align}
&P_{\vecX \vecY}\inp{\inb{(\vecX, \bs, \vecY)\in \cT^n_{X\bS Y}, D(P_{X\bS Y}||P_{X\bS}\times \bar{W})\geq \eta_2}}\nonumber\\
&=P_{\vecX \vecY}\inp{\cup_{\substack{P_{X\bS Y}\in \cP^n_{\cX\times\bar{\cS}\times\cY}:\\D(P_{X\bS Y}||P_{X\bS}\times \bar{W})\geq \eta_2}}\inb{(\vecX, \bs, \vecY)\in \cT^n_{X\bS Y} }}\nonumber\\
&=\sum_{\substack{P_{X\bS Y}\in \cP^n_{\cX\times\bar{\cS}\times\cY}:\\D(P_{X\bS Y}||P_{X\bS}\times \bar{W})\geq \eta_2}}P_{\vecX \vecY}\inp{(\vecX, \bs, \vecY)\in \cT^n_{X\bS Y} }.\nonumber
\end{align}
For any $P_{X\bS Y}\in \cP^n_{\cX\times\bar{\cS}\times\cY}$ such that $D(P_{X\bS Y}||P_{X\bS}\times \bar{W})\geq \eta_2$, we have
\begin{align}
&P_{\vecX \vecY}\inp{\inb{(\vecX, \bs, \vecY)\in \cT^n_{X\bS Y} }}\nonumber\\ 
&=\frac{1}{N}\sum_{\vecx_i\in \cT^n_{X|\bS}(\bs)}\sum_{\vecy\in \cT^n_{Y|X\bS}(\vecx_i, \bs)}W^n(\vecy|\vecx_i, \bs)\nonumber \\
&\stackrel{(a)}{\leq} \frac{1}{N}\sum_{\vecx_i\in \cT^n_{X|\bS}(\bs)}\exp\inb{-nD(P_{X\bS Y}||P_{X\bS}\times \bar{W})}\nonumber\\
&\leq \exp\inp{-n\eta_2}\nonumber
\end{align} where $(a)$ follows {from} \eqref{eq:type5}.
Thus, by \eqref{eq:type1},
\begin{align}
&P_{\vecX \vecY}\inp{\cE_2}\leq \sum_{\substack{P_{X\bS Y}\in \cP^n_{\cX\times\bar{\cS}\times\cY}:\\D(P_{X\bS Y}||P_{X\bS}\times \bar{W})\geq \eta_2}}\exp\inp{-n\eta_2}\nonumber\\
&\leq \exp\inp{-n\inp{\eta_2-\epsilon}} \text{ for large $n$ and $\epsilon>0$}.\label{app_eq:exponent2}
\end{align}
In order to evaluate the probability of $\cE_3$, let $\cP\subseteq \cP^n_{\cX\times \bar{\cS}\times \cY\times \cX}$ be such that for each $P_{X\bS YX'}\in \cP$ we have  $I(X;\bS)< \eta_1, D(P_{X\bS Y}||P_{X\bS}\times \bar{W})< \eta_2$,  $I(X';XY|\bar{S})\geq \eta_3$ and for some $S$ distributed over $\cS$, $I(X';S)< \delta, D(P_{X'SY}||P_{X'S}W)<\delta$ .
\begin{align}
&P_{\vecX \vecY}\inp{\cE_3}
\leq \sum_{P_{X\bS YX'}\in \cP}\frac{1}{N}\sum_{i: (\vecx_i, \vecx_j, \bs)\in \cT^n_{XX'\bS}\text{ for some }j\neq i} \sum_{\vecy\in \cT^n_{Y|X'X\bS}(\vecx_j, \vecx_i, \bs)}W^n(\vecy|\vecx_i, \bs)\label{eq_app:condition3_part}\\
&\leq \sum_{P_{X\bS YX'}\in \cP}\frac{1}{N}\inl{\inb{i: (\vecx_i, \vecx_j, \bs)\in \cT^n_{XX'\bS}\text{ for some }j\neq i}}\nonumber\\
&\stackrel{(a)}{\leq} \sum_{P_{X\bS YX'}\in \cP}\exp\inb{n\inp{\inps{R-I(X';\bS)}-I(X;X'\bS)+\epsilon/2}}\nonumber\\
&\stackrel{(b)}{\leq} \exp\inb{-n\inp{I(X;X'\bS)-\inps{R-I(X';\bS)}-\epsilon}}\label{app_eq:exponent31}
\end{align}
where $(a)$ follows from \eqref{eq:codebook3} and $(b)$ holds for large $n$. \eqref{eq_app:condition3_part} is also upper bounded by 
\begin{align}
&\sum_{P_{X\bS YX'}\in \cP}\frac{1}{N}\sum_{\vecx_i: \vecx_i\in \cT^n_{X|\bS}(\bs)}\sum_{\vecx_j\in \cT^n_{X'|X\bS}(\vecx_i, \bs)}\sum_{\vecy\in \cT^n_{Y|X'X\bS}(\vecx_j, \vecx_i, \bs)}W^n(\vecy|\vecx_i, \bs)\nonumber\\
&\stackrel{(a)}{\leq} \sum_{P_{X\bS YX'}\in \cP}\frac{1}{N}\sum_{\substack{\vecx_i:\\ \vecx_i\in \cT^n_{X|\bS}(\bs)}}\exp\inb{n\inp{\inps{R-I(X';X\bS)}+\epsilon}}(n+1)^{|\cY||\cX||\mathcal{\bS}|}\exp\inb{-nI(X';Y|X\bS)}\nonumber\\
&\stackrel{(b)}\leq \exp\inb{-n\inp{I(X';Y|X\bS)-\inps{R-I(X';X\bS)}-2\epsilon}}\label{app_eq:exponent32}
\end{align}
where $(a)$ follows from \eqref{eq:codebook1} and by noting that $\sum_{\vecy\in \cT^n_{Y|X'X\bS}(\vecx_j, \vecx_i, \bs)}W^n(\vecy|\vecx_i, \bs)\leq (n+1)^{|\cY||\cX||\mathcal{\bS}|}\exp\inp{-nI(X';Y|X\bS)}$. This is because $W^n(\vecy|\vecx_i, \bs)$ is the same for every $\vecy\in \cT^n_{Y|X\bS}(\vecx_i, \bs)$ and hence is upper bounded by $1/|\cT^n_{Y|X\bS}(\vecx_i, \bs)|\leq (n+1)^{|\cY||\cX||\mathcal{\bS}|}\exp\inp{-nH(Y|X\bS)}$ and $(b)$ holds for large $n$. The exponent in \eqref{eq:exponent_intermediate} follows from \eqref{app_eq:exponent1}, \eqref{app_eq:exponent2}, \eqref{app_eq:exponent31} and \eqref{app_eq:exponent32}.

Next, we show the exponent in Theorem~\ref{thm:pvt_rand}.

For $R\geq I(X;\bS)$, $A_1 = I(X;\bS)-\epsilon\geq \eta_1-\epsilon$. When $R<I(X;\bS)$, $A_1 = R-\epsilon$.   Next, we evaluate $A_2$. When $I(X;X'\bS)-\inps{R-I(X';\bS)}-\epsilon \geq t$ for some $t$ (TBD), $A_2\geq t$. Otherwise, when $I(X;X'\bS)-\inps{R-I(X';\bS)}\leq \epsilon + t$, we consider two cases. When $R\leq I(X';\bS)$, we have $I(X;X'|\bS)\leq I(X;X'\bS)\leq \epsilon + t$. Thus,  
\begin{align*}
&I(Y;X'|X\bS)-\inps{R-I(X';X\bS)}-2\epsilon\\
&=I(Y;X'|X\bS)-2\epsilon\\
& = I(YX;X'|\bS)-I(X;X'|\bS)-2\epsilon\\
&\geq \eta_3-t-3\epsilon \text{ because }I(YX;X'|\bS)>\eta_3.
\end{align*}
Thus, $A_2\geq \eta_3-t-3\epsilon$ in this case. When $R>I(X';S)$, 
\begin{align*}
R&\geq I(X;X'\bS)+I(X';\bS)-\epsilon - t\\
& \geq  I(X';X\bS) -\epsilon-t.
\end{align*}This implies that  $\inps{R-I(X';X\bS)}\leq R-I(X';X\bS)+\epsilon+t$. In this case,
\begin{align*}
&I(Y;X'|X\bS)-\inps{R-I(X';X\bS)}-2\epsilon\\
&\geq I(Y;X'|X\bS) -R+I(X';X\bS)-\epsilon-t-2\epsilon\\
&= I(X\bS Y;X')-R-t-3\epsilon\\
& = I(XY;X'|\bS)+I(X';\bS)- R - t - 3\epsilon\\
&\geq \eta_3 - R - t - 3\epsilon.
\end{align*}

With this, the exponent $\gamma$ \begin{align*}
\gamma &\geq \min\big\{\min\inb{\eta_1-\epsilon, R-\epsilon}, \eta_2-\epsilon,\\& \qquad \qquad  \max\inb{t, \min\inb{\eta_3 - t - \epsilon/4, \eta_3-R-t-3\epsilon}}\big\}.
\end{align*}
For $R = t = \eta_3/3$ and $\epsilon \rightarrow 0$ (note that $\epsilon>0$ may be arbitrarily small as long as $R\geq \epsilon$ as required by Lemma~\ref{app_lem:codebook}),  the exponent $\gamma$ can me made arbitrarily close to 
\begin{align}\min\inb{\eta_1, \eta_2,\eta_3/3}.\label{eq:final_exponent}
\end{align}
Next, we will show under Hypothesis $H_0$ too, the probability of error is arbitrarily small. Suppose the adversary's attack is $\vecs\in \cS^n$.
For each $\vecx_j\in \cC(P)$ and $\vecy\in \cY^n$, let $P_{\vecX'\vecY}(\vecx_j, \vecy) = \frac{1}{N}W^n(\vecy|\vecx_j, \vecs)$. Let $(\vecX', \vecY)\sim P_{\vecX'\vecY}$. Define $\tilde{\cE}_1 := \inb{(\vecX', \vecs)\in \cT^n_{X'S} \text{ such that } I(X';S)\geq \delta}$, $\tilde{\cE}_2 := \inb{(\vecX', \vecs, \vecY)\in \cT^n_{X'S Y} \text{ such that }  D(P_{X'S Y}||P_{X'S}\times {W})\geq \delta}$,  
$\tilde{\cE}_3 := \big\{(\vecX', \vecs, \vecY)\in \cT^n_{X'S Y} \text{ such that }  I(X';S)< \delta,\\ D(P_{X'S Y}||P_{X'S}\times {W})< \delta, \exists \vecx_i\neq \vecX' \text{ such that }(\vecx_i, \bs, \vecY)\in \cT^n_{X\bS Y} \text{ for some }\bs\in {\bar{\cS}}^n \text{ for which }I(X;\bS)< \eta_1 \\\text{ and }D(P_{X\bS Y}||P_{X\bS}\bar{W})<\eta_2, \text{ but } I(X;X'Y|{S})\geq \delta\big\}$, and 
$\tilde{\cE}_4 := \big\{\exists\bs\in {\bar{\cS}}^n \text{ such that }(\vecX', \bs, \vecs, \vecY)\in \cT^n_{X'\bS S Y} \\ \text{ for which }  I(X;\bS)< \eta_1,D(P_{X\bS Y}||P_{X\bS}\times \bar{W})< \eta_2, I(X';S)< \delta \text{ and }D(P_{X'SY}||P_{X'S}W)<\delta\big\}$.

These events are analogous to the events $\cE_1, \cE_2$, $\cE_3$ and $\cE_4$ defined under $H_1$, except that $(\eta_1, \eta_2, \eta_3)$ is exchanged with $(\delta, \delta, \delta)$. Following a similar line of argument, one can show that $P_{\vecX' \vecY}\inp{\tilde{\cE}_1\cup\tilde{\cE}_2\cup\tilde{\cE}_3\cup\tilde{\cE}_4}\leq \exp\inp{-n\delta/3}$ (see \eqref{eq:final_exponent}).

We will next argue that conditioned on the event $\tilde{\cE}_1^c\cap\tilde{\cE}_2^c\cap\tilde{\cE}_3^c\cap\tilde{\cE}^c_4$, the detector will not output $H_1$. 
This is because 
Definition~\ref{defn:eta_p} ensures that for $(\eta_1, \eta_2, \eta_3)\in \eta(P)$ and $\delta$ given by definition~\ref{defn:eta_p}, 
\begin{itemize}
	\item There does not exist $\vecx_i$, $\bar{\vecs}\in \bar{\cS}^n$ and  such that for $(\vecx_i, \vecX', \bar{\vecs}, \vecs, \vecY)\in \cT^n_{XX'\bar{S}SY}$, $I(X;\bar{S})< \eta_1$,  $D(P_{X\bar{S}Y}||P_{X\bar{S}}\bar{W})<\eta_2$, $I(X';S)< \delta$, $D(P_{X'SY}||P_{XS}W)<\delta$, and for $X\neq X'$,  $I(X';XY|\bar{S})<\eta_3$ and  $I(X;X'Y|S)<\delta$.
\end{itemize}
This implies that the error will happen only under $\tilde{\cE}_1\cup\tilde{\cE}_2\cup\tilde{\cE}_3\cup\tilde{\cE}_4$ which happens with probability at most $\exp\inp{-n\delta/3}$. This can be made arbitrarily small for large $n$.
\end{proof}


\begin{proof}[Proof of Lemma~\ref{app_lem:codebook}]\label{eq:proof_lem_codebook}
The proof of the lemma follows from the proof of \cite[Lemma 3]{csiszar1988capacity}. \eqref{eq:codebook1} is the same as \cite[eq. (3.1)]{csiszar1988capacity}. \eqref{eq:codebook2} can be obtained from the proof of \cite[eq. (3.2)]{csiszar1988capacity}, specifically by replacing $P_{X'S}$ with $P_{XS}$ and $\epsilon$ with $\epsilon/2$ in \cite[eq. (A8)]{csiszar1988capacity}. Equation \eqref{eq:codebook3} is obtained from the proof of \cite[eq. (3.3)]{csiszar1988capacity}, where for $a = (n+1)^{|\cX|}\exp\inb{n\inp{\inps{R-I(X';S)}-I(X;X'S)+\epsilon/4}}$, we choose $t = \exp\inb{n\inp{\inps{R-I(X';S)}-I(X;X'S)+\epsilon/2}}$. Note that for large enough $n$, $t> a\log{e}$ as required by \cite[eq. (A2)]{csiszar1988capacity}. 
\end{proof}

\newpage
\bibliographystyle{ieeetr}
\bibliography{refs}

\end{document}